\documentclass{elsarticle}
\usepackage{amsmath,amsthm,amssymb}
\usepackage{xspace,xypic,enumerate,color}
\usepackage{tikzfig}
\usepackage{bm}

 \tikzstyle{every picture}=[baseline=-0.25em,shorten <=-0.1pt]
\tikzstyle{dotpic}=[scale=0.5]
\tikzstyle{dotpicsm}=[scale=0.3]
\tikzstyle{dotpic2}=[scale=0.75]

\newcommand{\phantombox}[1]{\tikz[baseline=(current bounding box).east]{\path [use as bounding box] (0,0) rectangle #1;}}
\tikzstyle{braceedge}=[decorate,decoration={brace,amplitude=1mm,raise=-1mm}]
\tikzstyle{left hook arrow}=[left hook-latex]
\tikzstyle{right hook arrow}=[right hook-latex]
\tikzstyle{left label}=[label,anchor=east,xshift=1.5mm]
\tikzstyle{right label}=[label,anchor=west,xshift=-1.5mm]

\makeatletter
\newcommand{\boxshape}[3]{%
\pgfdeclareshape{#1}{
\inheritsavedanchors[from=rectangle] 
\inheritanchorborder[from=rectangle]
\inheritanchor[from=rectangle]{center}
\inheritanchor[from=rectangle]{north}
\inheritanchor[from=rectangle]{south}
\inheritanchor[from=rectangle]{west}
\inheritanchor[from=rectangle]{east}
\backgroundpath{
\southwest \pgf@xa=\pgf@x \pgf@ya=\pgf@y
\northeast \pgf@xb=\pgf@x \pgf@yb=\pgf@y

\@tempdima=#2
\@tempdimb=#3

\pgfpathmoveto{\pgfpoint{\pgf@xa - 5pt + \@tempdima}{\pgf@ya}}
\pgfpathlineto{\pgfpoint{\pgf@xa - 5pt - \@tempdima}{\pgf@yb}}
\pgfpathlineto{\pgfpoint{\pgf@xb + 5pt + \@tempdimb}{\pgf@yb}}
\pgfpathlineto{\pgfpoint{\pgf@xb + 5pt - \@tempdimb}{\pgf@ya}}
\pgfpathlineto{\pgfpoint{\pgf@xa - 5pt + \@tempdima}{\pgf@ya}}
\pgfpathclose
}
}}

\boxshape{NEbox}{0pt}{5pt}
\boxshape{SEbox}{0pt}{-5pt}
\boxshape{NWbox}{5pt}{0pt}
\boxshape{SWbox}{-5pt}{0pt}
\boxshape{EBox}{-3pt}{3pt}
\boxshape{WBox}{3pt}{-3pt}
\makeatother

\tikzstyle{cloud}=[shape=cloud,draw,minimum width=1.5cm,minimum height=1.5cm]

\tikzstyle{map}=[draw,shape=NEbox,inner sep=2pt,minimum height=6mm,fill=white]
\tikzstyle{dashedmap}=[draw,dashed,shape=NEbox,inner sep=2pt,minimum height=6mm,fill=white]
\tikzstyle{mapdag}=[draw,shape=SEbox,inner sep=2pt,minimum height=6mm,fill=white]
\tikzstyle{mapadj}=[draw,shape=SEbox,inner sep=2pt,minimum height=6mm,fill=white]
\tikzstyle{maptrans}=[draw,shape=SWbox,inner sep=2pt,minimum height=6mm,fill=white]
\tikzstyle{mapconj}=[draw,shape=NWbox,inner sep=2pt,minimum height=6mm,fill=white]

\tikzstyle{dot}=[inner sep=0.7mm,minimum width=0pt,minimum height=0pt,fill=black,draw,shape=circle]

\tikzstyle{black dot}=[dot]
\tikzstyle{white dot}=[dot,fill=white]
\tikzstyle{gray dot}=[dot,fill=gray!40!white]

\tikzstyle{alt white dot}=[white dot,label={[xshift=3mm,yshift=-0.05mm,font=\tiny]left:$*$}]
\tikzstyle{alt gray dot}=[gray dot,label={[xshift=3mm,yshift=-0.05mm,font=\tiny]left:$*$}]

\tikzstyle{white norm}=[rectangle,fill=white,draw=black,minimum height=2mm,minimum width=2mm,inner sep=0pt,font=\small]
\tikzstyle{gray norm}=[white norm,fill=gray!40!white]
\tikzstyle{black norm}=[white norm,fill=black]

\tikzstyle{arrs}=[-latex,font=\small,auto]
\tikzstyle{arrow plain}=[arrs]
\tikzstyle{arrow dashed}=[dashed,arrs]
\tikzstyle{arrow bold}=[very thick,arrs]
\tikzstyle{arrow hide}=[draw=white!0,-]
\tikzstyle{arrow reverse}=[latex-]
\tikzstyle{cdnode}=[]

\tikzstyle{wide point}=[fill=white,draw=black,shape=isosceles triangle,shape border rotate=90,isosceles triangle stretches=true,inner sep=1pt,minimum width=1.5cm,minimum height=5mm]
\tikzstyle{wide copoint}=[fill=white,draw=black,shape=isosceles triangle,shape border rotate=-90,isosceles triangle stretches=true,inner sep=1pt,minimum width=1.5cm,minimum height=4mm]
\tikzstyle{very wide copoint}=[fill=white,draw=black,shape=isosceles triangle,Shape border rotate=-90,isosceles triangle stretches=true,inner sep=1pt,minimum width=2.5cm,minimum height=4mm]
\tikzstyle{very wide empty copoint}=[draw=black,shape=isosceles triangle,shape border rotate=-90,isosceles triangle stretches=true,inner sep=1pt,minimum width=2.5cm,minimum height=4mm]
\tikzstyle{symm}=[ultra thick,shorten <=-1mm,shorten >=-1mm]

\tikzstyle{square box}=[rectangle,fill=white,draw=black,minimum height=5mm,minimum width=5mm,font=\small]
\tikzstyle{square gray box}=[rectangle,fill=gray!30,draw=black,minimum height=6mm,minimum width=6mm]
\tikzstyle{point}=[regular polygon,regular polygon rotate=180,regular polygon sides=3,draw=black,scale=0.75,inner sep=-0.5pt,minimum width=.6cm,fill=white]
\tikzstyle{copoint}=[regular polygon,regular polygon sides=3,draw=black,scale=0.75,inner sep=-0.5pt,minimum width=1cm,fill=white]
\tikzstyle{gray point}=[point,fill=gray!40!white]
\tikzstyle{gray copoint}=[copoint,fill=gray!40!white]

\tikzstyle{diredge}=[->]
\tikzstyle{rdiredge}=[<-]
\tikzstyle{dashed edge}=[dashed]

\tikzstyle{cross}=[preaction={draw=white, -, line width=3pt}]

\tikzstyle{black point}=[regular polygon,regular polygon rotate=180,regular polygon sides=3,draw=black,scale=0.25,inner sep=-0.5pt,minimum width=1cm,fill=black]

\tikzstyle{grey point}=[regular polygon,regular polygon rotate=180,regular polygon sides=3,draw=black,scale=0.25,inner sep=-0.5pt,minimum width=1cm,fill=gray!40!white]

\tikzstyle{white point}=[regular polygon,regular polygon rotate=180,regular polygon sides=3,draw=black,scale=0.25,inner sep=-0.5pt,minimum width=1cm,fill=white]

\tikzstyle{wide white point}=[regular polygon,regular polygon rotate=180,regular polygon sides=3,draw=black,xscale=0.25,yscale=0.15,inner sep=-0.5pt,minimum width=15mm,fill=white]

\tikzstyle{wide white copoint}=[regular polygon,regular polygon rotate=0,regular polygon sides=3,draw=black,xscale=0.25,yscale=0.15,inner sep=-0.5pt,minimum width=15mm,fill=white]

\tikzstyle{black copoint}=[regular polygon,regular polygon rotate=0,regular polygon sides=3,draw=black,scale=0.25,inner sep=-0.5pt,minimum width=1cm,fill=black]

\tikzstyle{grey copoint}=[regular polygon,regular polygon rotate=0,regular polygon sides=3,draw=black,scale=0.25,inner sep=-0.5pt,minimum width=1cm,fill=gray!40!white]

\tikzstyle{white copoint}=[regular polygon,regular polygon rotate=0,regular polygon sides=3,draw=black,scale=0.25,inner sep=-0.5pt,minimum width=1cm,fill=white]

\newcommand{\subsystemcolour}{gray!40!white}
\pgfkeys{/tikz/.cd, subsystemcolour/.store in=\subsystemcolour}
\tikzstyle{subsystem}=[postaction={decorate,decoration={markings,
    mark=at position .4 with {
    \draw[fill=\subsystemcolour, draw=black,scale=.1, rotate=-90] (0,0) -- (1,1.5) -- (-1,1.5) -- cycle;
    }}}]
\tikzstyle{None}=[circle, fill=white, inner sep=0pt]

\newcommand{\pantsalg}{%
\,\begin{tikzpicture}[dotpic,scale=0.8]
		\node [style=none] (0) at (-0.75, -0.5) {};
		\node [style=none] (1) at (0.25, 0.5) {};
		\node [style=none] (2) at (0.75, -0.5) {};
		\node [style=none] (3) at (0.25, -0.5) {};
		\node [style=none] (4) at (-0.25, 0.5) {};
		\node [style=none] (5) at (-0.25, -0.5) {};
		\draw [style=diredge, in=-90, out=90] (2.center) to (1.center);
		\draw [style=diredge, in=90, out=-90] (4.center) to (0.center);
		\draw [style=diredge, in=90, out=90, looseness=2.00] (5.center) to (3.center);
\end{tikzpicture}\,}

\newcommand{\dotcounit}[1]{%
\,\begin{tikzpicture}[dotpic,yshift=-1mm]
\node [#1] (a) at (0,0.25) {}; 
\draw [diredge] (0,-0.2)--(a);
\end{tikzpicture}\,}
\newcommand{\dotunit}[1]{%
\,\begin{tikzpicture}[dotpic,yshift=1.5mm]
\node [#1] (a) at (0,-0.25) {}; 
\draw [diredge] (a)--(0,0.2);
\end{tikzpicture}\,}
\newcommand{\dotcomult}[1]{%
\,\begin{tikzpicture}[dotpic,yshift=0.5mm]
	\node [#1] (a) {};
	\draw [diredge] (-90:0.55)--(a);
	\draw [diredge] (a) -- (45:0.6);
	\draw [diredge] (a) -- (135:0.6);
\end{tikzpicture}\,}
\newcommand{\dotmult}[1]{%
\,\begin{tikzpicture}[dotpic]
	\node [#1] (a) {};
	\draw [diredge] (a) -- (90:0.55);
	\draw [diredge] (a) (-45:0.6) -- (a);
	\draw [diredge] (a) (-135:0.6) -- (a);
\end{tikzpicture}\,}

\newcommand{\idpic}{%
\,\begin{tikzpicture}[dotpic]
  \draw[diredge] (0,-.5) to (0,.35);
\end{tikzpicture}\,}

\newcommand{\dotaction}[1]{%
\,\begin{tikzpicture}[dotpic,yshift=0.5mm]
	\node [#1] (a) {};
	\draw [diredge] (-90:0.55)--(a);
	\draw [diredge] (a) -- (45:0.6);
	\draw [rdiredge] (a) -- (135:0.6);
\end{tikzpicture}\,}
\newcommand{\dotcoaction}[1]{%
\,\begin{tikzpicture}[dotpic]
	\node [#1] (a) {};
	\draw [diredge] (a) -- (90:0.55);
	\draw [diredge] (a) (-45:0.6) -- (a);
	\draw [rdiredge] (a) (-135:0.6) -- (a);
\end{tikzpicture}\,}

\newcommand{\dotdualmult}[1]{%
\!\begin{tikzpicture}[dotpic]
		\node [style=white dot] (0) at (0, 0.3) {};
		\node [style=none] (1) at (-0.5, -0.4) {};
		\node [style=none] (2) at (0.5, -0.4) {};
		\node [style=none] (3) at (0, 0.8) {};
		\draw [style=diredge] (3.center) to (0);
		\draw [style=diredge, in=15, out=-30, looseness=1.50] (0) to (1.center);
		\draw [style=diredge, in=165, out=-150, looseness=1.50] (0) to (2.center);
\end{tikzpicture}\!}

\newcommand{\dotidualiser}[1]{%
\begin{tikzpicture}[dotpic,yshift=1.5mm]
	\node [#1] (a) {};
	\draw [diredge] (a) to (-90:0.5);
	\draw [diredge] (a) to (90:0.5);
\end{tikzpicture}}

\newcommand{\dotonly}[1]{%
\,\begin{tikzpicture}[dotpic,yshift=-0.3mm]
\node [#1] (a) at (0,0) {};
\end{tikzpicture}\,}

\newcommand{\dotnorm}[1]{%
\,\begin{tikzpicture}[dotpic,yshift=0.4mm]
		\node [style=none] (0) at (0, -0.4) {};
		\node [style=white norm] (1) at (0, -0) {};
		\node [style=none] (2) at (0, 0.5) {};
		\draw (0.center) to (1);
		\draw [style=diredge] (1) to (2.center);
\end{tikzpicture}\,}
\newcommand{\dotconorm}[1]{%
\,\begin{tikzpicture}[dotpic,yshift=0.4mm]
		\node [style=none] (0) at (0, -0.4) {};
		\node [style=white norm] (1) at (0, 0.1) {};
		\node [style=none] (2) at (0, 0.5) {};
		\draw [style=diredge] (1) to (0.center);
		\draw (2.center) to (1);
\end{tikzpicture}\,}

\newcommand{\blackunit}{\dotunit{dot}}

\newcommand{\blackmult}{\dotmult{dot}}

\newcommand{\whitedot}{\dotonly{white dot}}
\newcommand{\whiteunit}{\dotunit{white dot}}
\newcommand{\whitecounit}{\dotcounit{white dot}}
\newcommand{\whitemult}{\dotmult{white dot}}
\newcommand{\whitecomult}{\dotcomult{white dot}}

\newcommand{\whiteaction}{\dotaction{white dot}}
\newcommand{\whitecoaction}{\dotcoaction{white dot}}

\newcommand{\grayunit}{\dotunit{gray dot}}

\newcommand{\graymult}{\dotmult{gray dot}}

\newcommand{\subsystem}[3][gray!40!white]{\smash{\begin{aligned}\begin{tikzpicture}[dotpic,font=\scriptsize]
                \node at (.3,-.3) {$#2$};
                \node at (.3,.4) {$#3$};
		\draw [style=subsystem, subsystemcolour=#1] (0,-.5) to (0,.2);
                \draw (0,.2) to (0,.5);
\end{tikzpicture}\end{aligned}}}

\newcommand{\subsystemdag}[3][gray!40!white]{\smash{\begin{aligned}\begin{tikzpicture}[dotpic,font=\scriptsize,yscale=-1]
                \node at (.3,-.3) {$#2$};
                \node at (.3,.4) {$#3$};
		\draw [style=subsystem, subsystemcolour=#1] (0,-.5) to (0,.2);
                \draw (0,.2) to (0,.5);
\end{tikzpicture}\end{aligned}}}

\let\olddagger\dagger
\renewcommand{\dagger}{\ensuremath{\olddagger}\xspace}

\newcommand{\op}{\ensuremath{\textrm{\rm op}}}
\newcommand{\cat}[1]{\ensuremath{\mathbf{#1}}\xspace}
\newcommand{\tr}{\ensuremath{\textrm{\rm tr}}}
\newcommand{\FHilb}{\cat{FHilb}}

\newcommand{\Rel}{\cat{Rel}}

\newcommand{\CP}{\ensuremath{\mathrm{CP}\xspace}}
\newcommand{\CPs}{\ensuremath{\CP^*}\xspace}

\newcommand{\V}{\cat{V}}

\newcommand{\id}[1][]{\ensuremath{1_{#1}}}

\DeclareMathOperator{\Mor}{Mor}
\DeclareMathOperator{\dom}{dom}
\DeclareMathOperator{\cod}{cod}

\newcommand{\cheat}{\ensuremath{\textit{\it cheat}}}

\newcommand{\bra}[1]{\ensuremath{\langle #1 |}}
\newcommand{\ket}[1]{\ensuremath{| #1 \rangle}}
\newcommand{\braket}[2]{\ensuremath{\langle #1 | #2 \rangle}}
\newcommand{\ketbra}[2]{\ensuremath{| #1 \rangle \! \langle #2 |}}

\newcommand\sqrtD{\textrm{\footnotesize $\sqrt{D}$}\xspace}

\newcommand\oneoversqrtD{\ensuremath{{\textstyle{1\over{\sqrt{D}}}}}\xspace}

\theoremstyle{plain}
\newtheorem{theorem}{Theorem}[section]

\newtheorem{lemma}[theorem]{Lemma}
\newtheorem{proposition}[theorem]{Proposition}

\theoremstyle{definition}
\newtheorem{definition}[theorem]{Definition}

\newtheorem{example}[theorem]{Example}
\newtheorem{remark}[theorem]{Remark}

\tikzstyle{cdiag}=[matrix of math nodes, row sep=3em, column sep=3em, text height=1.5ex, text depth=0.25ex,inner sep=0.5em]
\tikzstyle{arrow above}=[transform canvas={yshift=0.5ex}]
\tikzstyle{arrow below}=[transform canvas={yshift=-0.5ex}]

\newcommand{\vkcleararrows}{%
\tikzstyle{vkarrow1}=[arrow plain]
\tikzstyle{vkarrow2}=[arrow plain]
\tikzstyle{vkarrow3}=[arrow plain]
\tikzstyle{vkarrow4}=[arrow plain]
\tikzstyle{vkarrow5}=[arrow plain]
\tikzstyle{vkarrow6}=[arrow plain]
\tikzstyle{vkarrow7}=[arrow plain]
\tikzstyle{vkarrow8}=[arrow plain]
\tikzstyle{vkarrow9}=[arrow plain]
\tikzstyle{vkarrow10}=[arrow plain]
\tikzstyle{vkarrow11}=[arrow plain]
\tikzstyle{vkarrow12}=[arrow plain]}
\vkcleararrows

\usepackage{color}
\def\bR{\begin{color}{red}} 
\def\bB{\begin{color}{blue}}
\def\bM{\begin{color}{magenta}}
\def\bC{\begin{color}{cyan}}
\def\bW{\begin{color}{white}}
\def\bBl{\begin{color}{black}} 
\def\bG{\begin{color}{green}}
\def\bY{\begin{color}{yellow}}
\def\e{\end{color}\xspace}

\usepackage{tikzfig}
\begin{document}

\title{The CBH characterisation theorem\\beyond algebraic quantum theory}
\author[ch]{Chris Heunen\fnref{fn1}}
\ead{chris.heunen@ed.ac.uk}
\author[ak]{Aleks Kissinger}
\ead{aleks.kissinger@cs.ox.ac.uk}
\address[ch]{University of Edinburgh, UK}
\address[ak]{University of Oxford, UK}
\fntext[fn1]{Supported by EPSRC Fellowship EP/L002388/1.}

\begin{abstract}
  The CBH theorem characterises quantum theory within a C*-algebraic framework.
  Namely, mathematical properties of C*-algebras modelling quantum systems are equivalent to constraints that are information-theoretic in nature:
  (1) noncommutativity of subalgebras is equivalent to impossibility of signalling;
  (2) noncommutativity of the whole algebra is equivalent to impossibility of broadcasting; 
  (3) the existence of entangled states is implied by the impossibility of secure bit commitment (with the converse conjectured). 
  However, the C*-algebraic framework has drawn criticism as it already contains much of the mathematical structure of quantum theory such as complex linearity. 
  We address this issue by a generalising C*-algebras categorically. 
  In this framework, equivalence (1) holds, equivalence (2) becomes a strict implication, and implication (3) fails in general.
  Thus we identify exactly what work is being done by the complex-linear structure of  C*-algebras.
  In doing so, we uncover a richer hierarchy of notions of `classicality' and `quantumness' of information than visible in the concrete case.
\end{abstract}
\maketitle

\section{Introduction}

Does information play a significant role in the foundations of physics? This question, often abbreviated `it from bit' after John Wheeler, has received significant attention, and lies at the root of quantum information theory. The seminal work by Clifton, Bub, and Halvorson~\cite{cliftonbubhalvorson} isolates quantum theory according to the following three information-theoretic constraints:
\begin{itemize}
\item It is impossible to \emph{signal} information faster than light;
\item It is impossible to \emph{broadcast} an unknown state;
\item It is impossible to securely implement \emph{bit commitment};
\end{itemize}
by linking each of them, respectively, with the following algebraic conditions, characteristic of quantum theory:
\begin{itemize}
\item Distinct systems are \emph{kinematically independent};
\item There exist \emph{noncommuting} observables;
\item There exist \emph{entangled}, or \emph{nonlocal}, states.
\end{itemize}
The first two pairs of properties are proven equivalent; for the third only one implication is proven, and the other conjectured.

However, a criticism often raised against this result is that a \mbox{C*-algebraic} framework, including complex numbers and linearity, is assumed from the start~\cite{spekkens:toytheory,barrett:gpt,barnumetal:nobroadcasting}. In the words of one of the authors himself~\cite[page~204]{schlosshauer:enigma}:
\begin{quote}
  The characterization theorem we proved assumes a C*-algebraic framework for physical theories, which I would now regard as not sufficiently general in the relevant sense, even though it includes a broad class of classical and quantum theories, including field theories, and hybrid theories with superselection rules.
\end{quote}

Recent axiomatisations \cite{hardy:reconstruction,masanes2011derivation,masanes2013existence,hoehn2015quantum,Chiri2} have to a great extent avoided this problem by starting with a framework which retains only the convex structure of probabilistic states. Nevertheless, these frameworks maintain some remnant of linearity in the form of convexity, which does a great deal of work.

In this paper, we investigate the issue by generalising quantum theory in a different direction, which instead retains only the algebraic structure of interaction between classical and quantum systems, via the \emph{CP*-construction}~\cite{coeckeheunenkissinger:cpstar}, which enables us to consider abstract C*-algebras in a wide variety of categories beyond the usual examples based on vector spaces. By applying this construction to the category of Hilbert spaces, we obtain algebraic quantum information theory in the usual sense, where objects are (finite-dimensional) C*-algebras and morphisms are completely positive linear maps. However, if we apply this to other categories, we obtain nonstandard models of quantum theory. For example, applying the construction to the category of relations gives a \textit{possibilistic} model of quantum theory, where probabilities of measurement outcomes are replaced by boolean-valued possibilities (`possible' vs.\ `impossible'), in a similar spirit to \textit{possibilistic physics} as described by Fritz~\cite{fritz2009possibilistic} or Schumacher and Westmoreland's \textit{modal quantum theory}~\cite{schumacher2012modal}.


We phrase the above implications in categorical terms, and show which of them survive and which ones fail. That is, our main contribution is to prove the following (non)implications in this generalised setting:

\begin{align*}
  \textit{information theory} & \hspace*{7mm} \textit{quantum theory} \\
  \text{no signalling} & \;\Leftrightarrow\;
  \text{kinematic independence} \\[-1mm]
  \text{no broadcasting} &
  \begin{aligned}\begin{array}{c}
    \Rightarrow \\[-1.25ex]
    \nLeftarrow
  \end{array}\end{aligned}
  \text{noncommutativity} \\[-1mm]
  \text{no bit commitment} &
  \begin{aligned}\begin{array}{c}
    \nLeftarrow
  \end{array}\end{aligned}
  \text{entanglement}
\end{align*}

The first equivalence of no-signalling with kinematic independence indeed lifts almost unmodified to any category arising from the CP*-construction. However, in the second case, the commutativity of the algebra of observables merely implies the existence of a broadcasting map (and not vice-versa). Interestingly, this yields finer-grained \emph{notions of classicality} for the systems in a theory. In particular, we can combine this with the result of~\cite{CompQL} to show that, in the case of possibilistic quantum theory, this yields a beautiful group-theoretic hierarchy of classical systems. Rather than requiring the existence of a broadcasting map, we can force a system to be classical by requiring that the algebra associated with the system is commutative or that the partially ordered set of \emph{propositions} (\textit{i.e.}\ the CP*-generalisation of the lattice of projections) is distributive, or both. In the case of quantum theory, either of these is equivalent to the existence of a broadcasting map. However, in possibilistic quantum theory, the following inclusions are strict:
\[ \textit{distributive \& commutative} \subsetneq \textit{commutative} \subsetneq \textit{broadcasting} \]
The systems in possibilistic quantum theory are groupoids~\cite{heunencontrerascattaneo:groupoids}, which are a generalisation of groups where the multiplication is allowed to be partially-defined. We show that broadcasting, the weakest `notion of classicality', implies that these groupoids are in fact disjoint unions of groups. Then, stricter notions of classicality yield smaller classes of groups. The `classicality' hierarchy above corresponds exactly to groupoids arising as disjoint unions of:
\[ \textit{locally cyclic groups} \subsetneq \textit{abelian groups} \subsetneq \textit{groups} \]

For the case of bit commitment, we first assume that our system is `quantum enough' to admit secure bit commitment. Namely, we assume that it is described by a noncommutative algebra, which by the second implication is strictly weaker than assuming it is no-broadcasting. We then exhibit a noncommutative system in $\CPs[\Rel]$ which simultaneously admits entanglement and a secure bit commitment protocol. The converse remains an open question.

A notable feature of this work is the convergence of two previously disconnected threads in the foundational study of quantum information. On the one hand, we use many abstract, diagrammatic techniques that originated within the program of \textit{categorical quantum mechanics}~\cite{abramskycoecke:categoricalsemantics,heunenvicary:cqt,PQP} (CQM). Yet in contrast to prior CQM work, which relies heavily on the tensor product to give a notion of `subsystem, we develop abstract, categorical analogues to subsystems which are more general than tensor factors, and develop concepts such as signalling and broadcasting between such subsystems in the style of \textit{algebraic quantum information}~\cite{keyl2002fundamentals} (AQI). We hope this methodology, which both enables AQI-style generalised systems and CQM-style diagrammatic manipulations, will find many new applications in areas such as algebraic quantum field theory.

\begin{remark}
  Note that we refer to the condition called `nonlocality' in  \cite{cliftonbubhalvorson} simply as `entanglement'. This is to avoid a terminology clash with the related (and now more common) use of the term nonlocality, namely the absence of a locally realistic model for observed correlations.
\end{remark}

The rest of this article is laid out as follows.
Section~\ref{sec:cpstar} sets up our general categorical framework.
Sections~\ref{sec:signalling}--\ref{sec:commitment} then investigate one equivalence each: Section~\ref{sec:signalling} signalling, Section~\ref{sec:broadcasting} broadcasting, and Section~\ref{sec:commitment} bit commitment.
We thank Katriel Cohn--Gordon and Mariami Gachechiladze for their MSc thesis work, supervised by Chris Heunen, on parts of this topic~\cite{cohngordon:bitcommitment,gachechiladze:heisenberg}.

\section{The CP*-construction}\label{sec:cpstar}

In this section, we will briefly introduce compact dagger categories and the CP*-construction. We assume familiarity with basic categorical concepts, notably the notion of a \emph{symmetric monoidal category} (see e.g.~\cite{maclane:categories}).

The \emph{CP*-construction} transforms one symmetric monoidal category into another in a way that mirrors the passage from finite-dimensional Hilbert spaces and linear maps to finite-dimensional C*-algebras and completely positive linear maps. It operates on a certain type of symmetric monoidal category called a \emph{compact dagger category}, whose definition we now recall.

\begin{definition}\label{def:compact-closed}
  A \emph{compact closed category} is a symmetric monoidal category $\cat{C}$ such that every object $A$ in $\cat{C}$ has a \emph{dual} object $A^*$ and morphisms
  \begin{align}
  \varepsilon \colon A \otimes A^* & \to I
  & \eta & \colon I \to A^* \otimes A \notag
  \\ \intertext{satisfying}
  (\varepsilon \otimes \id) \circ (\id \otimes \eta) & = \id[A]
  & \id[A^*] & = (\id \otimes \varepsilon) \circ (\eta \otimes \id).
  \label{eq:compactstruc}
  \end{align}
\end{definition}

We will represent morphisms in such a category using string diagrams. For more information we refer to the survey~\cite{SelingerSurvey}. 
Objects are depicted as labelled wires with upward directed arrows, and their duals as downward wires:
\ctikzfig{directed-wires}
The morphisms $\varepsilon$ and $\eta$ are called \emph{caps} and \emph{cups}, drawn as:
\ctikzfig{directed_cap_cup}
Using this notation, the equations from~\eqref{eq:compactstruc} become:
\ctikzfig{line_yank_with_dual}
These `snake equations' embody the relationship between a maximally entangled state (e.g. the Bell state) and the associated effect arising from a Bell measurement. As such, they represent the key to quantum teleportation~\cite{abramskycoecke:categoricalsemantics}.

We can furthermore define the (partial) trace of a morphism $f : X \otimes A \to X \otimes B$ in a compact closed category via cups and caps:
\[ \tr_X(f) \ =\ \begin{tikzpicture}[dotpic]
	\begin{pgfonlayer}{nodelayer}
		\node [style=square box, minimum width=1.2cm] (0) at (0, 0) {$f$};
		\node [style=none] (1) at (-0.5, 0.5) {};
		\node [style=none] (2) at (-2, 0.5) {};
		\node [style=none] (3) at (0.5, 1.75) {};
		\node [style=label] (4) at (-1.25, 1.75) {$X$};
		\node [style=none] (5) at (0.5, 0.5) {};
		\node [style=none] (6) at (-1.25, -1.25) {};
		\node [style=none] (7) at (0.5, 1) {};
		\node [style=none] (8) at (0.5, -0.5) {};
		\node [style=none] (9) at (-2, -0.5) {};
		\node [style=none] (10) at (-0.5, -0.5) {};
		\node [style=none] (11) at (0.5, -1.75) {};
		\node [style=none] (12) at (-1.25, 1.25) {};
		\node [style=none] (13) at (0.5, -1) {};
		\node [style=right label] (14) at (0.75, 1.5) {$B$};
		\node [style=right label] (15) at (0.75, -1.5) {$A$};
	\end{pgfonlayer}
	\begin{pgfonlayer}{edgelayer}
		\draw [style=diredge, in=0, out=90, looseness=1.00] (1.center) to (12.center);
		\draw [in=90, out=180, looseness=1.00] (12.center) to (2.center);
		\draw [style=diredge] (5.center) to (7.center);
		\draw (7.center) to (3.center);
		\draw [in=-90, out=0, looseness=1.00] (6.center) to (10.center);
		\draw [style=diredge, in=180, out=-90, looseness=1.00] (9.center) to (6.center);
		\draw (13.center) to (8.center);
		\draw [style=diredge] (11.center) to (13.center);
		\draw (2.center) to (9.center);
	\end{pgfonlayer}
\end{tikzpicture} \]
where:
\ctikzfig{other-cap_cup}
(and similarly for cups).

An important property of compact closed categories is that two morphisms are equal whenever their string diagrams can be deformed into one another. In other words, the only relevant data in a string diagram is its connectivity~\cite{SelingerSurvey}.

\begin{definition}
A \emph{dagger category} is a category equipped with a contravariant functor $(-)^\dagger \colon \cat C^\op \to \cat C$ that satisfies $A^\dag=A$ on objects and $f^{\dag\dag}=f$ on morphisms.
\end{definition}

Dagger categories enable us to formulate abstract versions of many familiar linear-algebraic concepts. Notably, a morphism in a dagger category is called \textit{self-adjoint} if $f^\dag = f$ and \textit{positive} if there exists $g$ such that $f = g^\dag \circ g$. An isomorphism in a dagger category is \emph{unitary} if $f^{-1} = f^\dagger$.

\begin{definition}
A \emph{compact dagger category} is a category that is both a compact closed category and a dagger category, such that the coherence isomorphisms (associators, unitors, and swap maps) are unitary, and $\varepsilon_{A}^\dagger = \eta_{A^*}$.
\end{definition}

In a compact dagger category, a morphism can take four forms, the morphism itself and its \emph{adjoint}:
\[\begin{tikzpicture}[dotpic]
	\begin{pgfonlayer}{nodelayer}
		\node [style=square box] (0) at (0, 0) {$f$};
		\node [style=none] (1) at (0, -1.5) {};
		\node [style=none] (2) at (0, 1.5) {};
		\node [style=right label] (3) at (0.25, 1.25) {$B$};
		\node [style=right label] (4) at (0.25, -1.25) {$A$};
	\end{pgfonlayer}
	\begin{pgfonlayer}{edgelayer}
		\draw [style=diredge] (0) to (2.center);
		\draw [style=diredge] (1.center) to (0);
	\end{pgfonlayer}
\end{tikzpicture} \qquad\qquad\qquad \begin{tikzpicture}[dotpic]
	\begin{pgfonlayer}{nodelayer}
		\node [style=square box] (0) at (0, 0) {$f^\dagger$};
		\node [style=right label] (1) at (0.25, -1.25) {$B$};
		\node [style=none] (2) at (0, 1.5) {};
		\node [style=right label] (3) at (0.25, 1.25) {$A$};
		\node [style=none] (4) at (0, -1.5) {};
	\end{pgfonlayer}
	\begin{pgfonlayer}{edgelayer}
		\draw [style=diredge] (0) to (2.center);
		\draw [style=diredge] (4.center) to (0);
	\end{pgfonlayer}
\end{tikzpicture} \]
as well as its \emph{transpose} and \emph{conjugate}:
\[ \begin{tikzpicture}[dotpic]
	\begin{pgfonlayer}{nodelayer}
		\node [style=square box] (0) at (0, 0) {$f^*$};
		\node [style=none] (1) at (0, 1.5) {};
		\node [style=none] (2) at (0, -1.5) {};
		\node [style=right label] (3) at (0.25, -1.25) {$B$};
		\node [style=right label] (4) at (0.25, 1.25) {$A$};
	\end{pgfonlayer}
	\begin{pgfonlayer}{edgelayer}
		\draw [style=diredge] (0) to (2.center);
		\draw [style=diredge] (1.center) to (0);
	\end{pgfonlayer}
\end{tikzpicture} =\ \begin{tikzpicture}[dotpic]
	\begin{pgfonlayer}{nodelayer}
		\node [style=square box] (0) at (0, 0) {$f$};
		\node [style=none] (1) at (1.5, 0.5) {};
		\node [style=none] (2) at (1.5, -1.25) {};
		\node [style=none] (3) at (1.5, -0.5) {};
		\node [style=none] (4) at (1.5, -0.5) {};
		\node [style=none] (5) at (1.5, 0.5) {};
		\node [style=none] (6) at (-1.5, 1.25) {};
		\node [style=none] (7) at (-1.5, 0.5) {};
		\node [style=none] (8) at (-1.5, -0.5) {};
		\node [style=none] (9) at (-1.5, -0.5) {};
		\node [style=none] (10) at (0, 0.5) {};
		\node [style=none] (11) at (0, -0.5) {};
		\node [style=right label] (12) at (-1.25, 1) {$A$};
		\node [style=right label] (13) at (1.75, -1.25) {$B$};
	\end{pgfonlayer}
	\begin{pgfonlayer}{edgelayer}
		\draw [in=-90, out=90, looseness=1.00] (2.center) to (3.center);
		\draw [style=diredge] (5.center) to (4.center);
		\draw [style=diredge, in=90, out=-90, looseness=1.00] (6.center) to (7.center);
		\draw (7.center) to (8.center);
		\draw [style=swap, bend right=90, looseness=1.50] (9.center) to (11.center);
		\draw [style=swap] (11.center) to (0);
		\draw [style=swap] (0) to (10.center);
		\draw [style=swap, bend left=90, looseness=1.50] (10.center) to (5.center);
	\end{pgfonlayer}
\end{tikzpicture}
  \qquad\quad
   \begin{tikzpicture}[dotpic]
	\begin{pgfonlayer}{nodelayer}
		\node [style=square box] (0) at (0, 0) {$f_*$};
		\node [style=none] (1) at (0, 1.5) {};
		\node [style=none] (2) at (0, -1.5) {};
		\node [style=right label] (3) at (0.25, 1.25) {$B$};
		\node [style=right label] (4) at (0.25, -1.25) {$A$};
	\end{pgfonlayer}
	\begin{pgfonlayer}{edgelayer}
		\draw [style=diredge] (0) to (2.center);
		\draw [style=diredge] (1.center) to (0);
	\end{pgfonlayer}
\end{tikzpicture} =\ \begin{tikzpicture}[dotpic]
	\begin{pgfonlayer}{nodelayer}
		\node [style=square box] (0) at (0, 0) {$f^\dag$};
		\node [style=none] (1) at (1.5, 0.5) {};
		\node [style=none] (2) at (1.5, -1.25) {};
		\node [style=none] (3) at (1.5, -0.5) {};
		\node [style=none] (4) at (1.5, -0.5) {};
		\node [style=none] (5) at (1.5, 0.5) {};
		\node [style=none] (6) at (-1.5, 1.25) {};
		\node [style=none] (7) at (-1.5, 0.5) {};
		\node [style=none] (8) at (-1.5, -0.5) {};
		\node [style=none] (9) at (-1.5, -0.5) {};
		\node [style=none] (10) at (0, 0.5) {};
		\node [style=none] (11) at (0, -0.5) {};
		\node [style=right label] (12) at (1.75, -1.25) {$A$};
		\node [style=right label] (13) at (-1.25, 1) {$B$};
	\end{pgfonlayer}
	\begin{pgfonlayer}{edgelayer}
		\draw [in=-90, out=90, looseness=1.00] (2.center) to (3.center);
		\draw [style=diredge] (5.center) to (4.center);
		\draw [style=diredge, in=90, out=-90, looseness=1.00] (6.center) to (7.center);
		\draw (7.center) to (8.center);
		\draw [style=swap, bend right=90, looseness=1.50] (9.center) to (11.center);
		\draw [style=swap] (11.center) to (0);
		\draw [style=swap] (0) to (10.center);
		\draw [style=swap, bend left=90, looseness=1.50] (10.center) to (5.center);
	\end{pgfonlayer}
\end{tikzpicture}
\]

\begin{example}
  The category \FHilb of finite-dimensional Hilbert spaces and linear maps forms a compact dagger category. The dagger is given by the linear algebraic adjoint:
  \[ \braket{f^\dagger(u)}{v} = \braket{u}{f(v)} \]
  The dual $H^*$ is the linear algebraic dual, and cups and caps are given respectively as:
  \[ \varepsilon(\ket{\psi} \otimes \bra{\phi}) = \braket{\phi}{\psi} \qquad\qquad \eta = \sum \bra{\phi_i} \otimes \ket{\phi_i} \]
  where $\eta$ does not depend on the choice of orthonormal basis $\{ \ket{\phi_i} \}_i$.
  From this it follows that $f^* : K^* \to H^*$ is given by pre-composition:
  \[ f^*(\bra{\xi}) = \bra{\xi} \circ f \]
  which is sometimes called `operator transpose'. If we fix a bases for $K$ and $H$, we can identify them with their duals, in which case $f^*$ corresponds to matrix transposition.
\end{example}

\begin{example}
  The category \Rel of sets and relations also forms a compact dagger category. The composition of $R \subseteq A \times B$ and $S \subseteq B \times C$ is the usual composition of relations:
  \[ (a,c) \in S \circ R \iff \exists b \in B \colon (a,b) \in R, (b,c) \in  S \]
  and the monoidal product is given by cartesian product:
  \[ A \otimes B = A \times B \]
  \[ ((a,b), (c,d)) \in R \otimes S \iff (a,c) \in R, (b,d) \in S \]
  Hence the monoidal unit $I = \{*\}$ is a single-element set. The dagger is given by relational converse:
  \[ (b,a) \in R^\dagger \iff (a,b) \in R \]
  and cups and caps are:
  \[ \eta = \{ (*, (a,a)) \mid a \in A \} \quad\qquad
     \varepsilon = \{ ((a,a), *) \mid a \in A \} \]
\end{example}

The CP*-construction lets us build a new category whose objects are abstract C*-algebras coming from the original category and whose morphisms are abstract completely positive maps. To do this, we first make precise what we mean by `abstract C*-algebra'.

\begin{definition}
  A \emph{monoid} in a compact dagger category is an object $A$ together with a morphism $\whitemult \colon A \otimes A \to A$ and $\whiteunit \colon I \to A$ satisfying
  \[ \begin{tikzpicture}[dotpicsm]
	\begin{pgfonlayer}{nodelayer}
		\node [style=none] (0) at (1.75, -0.25) {};
		\node [style=none] (1) at (1, -2.25) {};
		\node [style=white dot] (2) at (3, 1) {};
		\node [style=none] (3) at (-1.75, -0.25) {};
		\node [style=none] (4) at (-3, 2.5) {};
		\node [style=none] (5) at (-4.25, -0.25) {};
		\node [style=white dot] (6) at (-1.75, -0.75) {};
		\node [style=none] (7) at (0, 0) {$=$};
		\node [style=none] (8) at (4.25, -0.25) {};
		\node [style=white dot] (9) at (1.75, -0.75) {};
		\node [style=none] (10) at (3, 2.5) {};
		\node [style=none] (11) at (-1.75, -0.25) {};
		\node [style=none] (12) at (-4.25, -2.25) {};
		\node [style=none] (13) at (-4.25, -0.25) {};
		\node [style=none] (14) at (2.5, -2.25) {};
		\node [style=none] (15) at (-1, -2.25) {};
		\node [style=none] (16) at (-2.5, -2.25) {};
		\node [style=none] (17) at (4.25, -0.25) {};
		\node [style=white dot] (18) at (-3, 1) {};
		\node [style=none] (19) at (4.25, -2.25) {};
		\node [style=none] (20) at (1.75, -0.25) {};
	\end{pgfonlayer}
	\begin{pgfonlayer}{edgelayer}
		\draw [style=diredge, in=-165, out=90, looseness=1.00] (5.center) to (18);
		\draw [style=diredge, in=-15, out=90, looseness=1.00] (3.center) to (18);
		\draw [style=diredge] (18) to (4.center);
		\draw [style=diredge, in=-165, out=90, looseness=1.00] (16.center) to (6);
		\draw [style=diredge, in=-15, out=90, looseness=1.00] (15.center) to (6);
		\draw [style=swap] (6) to (11.center);
		\draw [style=swap] (12.center) to (13.center);
		\draw [style=diredge, in=-15, out=90, looseness=1.00] (17.center) to (2);
		\draw [style=diredge, in=-165, out=90, looseness=1.00] (0.center) to (2);
		\draw [style=diredge] (2) to (10.center);
		\draw [style=diredge, in=-15, out=90, looseness=1.00] (14.center) to (9);
		\draw [style=diredge, in=-165, out=90, looseness=1.00] (1.center) to (9);
		\draw [style=swap] (9) to (20.center);
		\draw [style=swap] (19.center) to (8.center);
	\end{pgfonlayer}
\end{tikzpicture} \qquad\qquad \begin{tikzpicture}[dotpic]
	\begin{pgfonlayer}{nodelayer}
		\node [style=white dot] (0) at (-3, 0.5) {};
		\node [style=none] (1) at (-2, -0.5) {};
		\node [style=none] (2) at (-4, -0.5) {};
		\node [style=none] (3) at (-3, 1.5) {};
		\node [style=white dot] (4) at (-4, -1) {};
		\node [style=none] (5) at (-2, -1.5) {};
		\node [style=none] (6) at (-1, -0) {$=$};
		\node [style=none] (7) at (0, 1.5) {};
		\node [style=none] (8) at (0, -1.5) {};
		\node [style=white dot] (9) at (4, -1) {};
		\node [style=none] (10) at (3, 1.5) {};
		\node [style=none] (11) at (2, -0.5) {};
		\node [style=none] (12) at (4, -0.5) {};
		\node [style=white dot] (13) at (3, 0.5) {};
		\node [style=none] (14) at (2, -1.5) {};
		\node [style=none] (15) at (1, -0) {$=$};
	\end{pgfonlayer}
	\begin{pgfonlayer}{edgelayer}
		\draw [style=diredge, in=-165, out=90, looseness=1.00] (2.center) to (0);
		\draw [style=diredge, in=-15, out=90, looseness=1.00] (1.center) to (0);
		\draw [style=diredge] (0) to (3.center);
		\draw [style=swap] (4) to (2.center);
		\draw [style=swap] (5.center) to (1.center);
		\draw [style=diredge] (8.center) to (7.center);
		\draw [style=diredge, in=-15, out=90, looseness=1.00] (12.center) to (13);
		\draw [style=diredge, in=-165, out=90, looseness=1.00] (11.center) to (13);
		\draw [style=diredge] (13) to (10.center);
		\draw [style=swap] (9) to (12.center);
		\draw [style=swap] (14.center) to (11.center);
	\end{pgfonlayer}
\end{tikzpicture}. \]
  A \emph{dagger Frobenius structure} is a monoid satisfying
  \[ \begin{tikzpicture}[dotpic]
	\begin{pgfonlayer}{nodelayer}
		\node [style=white dot] (0) at (-4.5, -0.75) {};
		\node [style=none] (1) at (-3.5, 1.5) {};
		\node [style=none] (2) at (-4.5, -1.5) {};
		\node [style=none] (3) at (-2.25, 0) {$=$};
		\node [style=white dot] (4) at (-5.75, 0.75) {};
		\node [style=none] (5) at (-6.75, -1.5) {};
		\node [style=none] (6) at (-5.75, 1.5) {};
		\node [style=none] (7) at (4.5, -1.5) {};
		\node [style=white dot] (8) at (5.75, 0.5) {};
		\node [style=none] (9) at (3.5, 1.5) {};
		\node [style=white dot] (10) at (4.5, -0.75) {};
		\node [style=none] (11) at (5.75, 1.5) {};
		\node [style=none] (12) at (6.75, -1.5) {};
		\node [style=none] (13) at (-1, -1.5) {};
		\node [style=none] (14) at (1, -1.5) {};
		\node [style=white dot] (15) at (0, -0.5) {};
		\node [style=none] (16) at (-1, 1.5) {};
		\node [style=none] (17) at (1, 1.5) {};
		\node [style=white dot] (18) at (0, 0.5) {};
		\node [style=none] (19) at (2.25, 0) {$=$};
	\end{pgfonlayer}
	\begin{pgfonlayer}{edgelayer}
		\draw [style=diredge, in=-90, out=15, looseness=1.00] (0) to (1.center);
		\draw [style=diredge] (2.center) to (0);
		\draw [style=diredge, in=-165, out=90, looseness=1.00] (5.center) to (4);
		\draw [style=diredge] (4) to (6.center);
		\draw [style=diredge, in=-15, out=165, looseness=1.25] (0) to (4);
		\draw [style=diredge, in=-90, out=165, looseness=1.00] (10) to (9.center);
		\draw [style=diredge] (7.center) to (10);
		\draw [style=diredge, in=-15, out=90, looseness=1.00] (12.center) to (8);
		\draw [style=diredge] (8) to (11.center);
		\draw [style=diredge, in=-165, out=15, looseness=1.25] (10) to (8);
		\draw [style=diredge, in=-15, out=90, looseness=1.00] (14.center) to (15);
		\draw [style=diredge, in=-165, out=90, looseness=1.00] (13.center) to (15);
		\draw [style=diredge, in=-90, out=15, looseness=1.00] (18) to (17.center);
		\draw [style=diredge, in=-90, out=165, looseness=1.00] (18) to (16.center);
		\draw [style=diredge] (15) to (18);
	\end{pgfonlayer}
\end{tikzpicture}, \]
  where $\whitecomult = (\whitemult)^\dag$ and $\whitecounit = (\whiteunit)^\dag$.
  It is \emph{symmetric} when
  \begin{equation}\label{eq:symmetric}
    \begin{tikzpicture}[dotpic]
	\begin{pgfonlayer}{nodelayer}
		\node [style=white dot] (0) at (2.25, 1) {};
		\node [style=white dot] (1) at (2.25, 0) {};
		\node [style=none] (2) at (-3.25, -0.25) {};
		\node [style=none] (3) at (-1.25, -0.25) {};
		\node [style=none] (4) at (-1.25, -1.5) {};
		\node [style=none] (5) at (3.25, -1.5) {};
		\node [style=white dot] (6) at (-2.25, 0.5) {};
		\node [style=white dot] (7) at (-2.25, 1.5) {};
		\node [style=none] (8) at (0, 0) {$=$};
		\node [style=none] (9) at (1.25, -1.5) {};
		\node [style=none] (10) at (-3.25, -1.5) {};
	\end{pgfonlayer}
	\begin{pgfonlayer}{edgelayer}
		\draw [style=diredge, in=-165, out=90, looseness=1.00] (9.center) to (1);
		\draw [style=diredge, in=-15, out=90, looseness=1.00] (5.center) to (1);
		\draw [style=diredge] (1) to (0);
		\draw [style=diredge, in=-165, out=90, looseness=1.00] (2.center) to (6);
		\draw [style=diredge, in=-15, out=90, looseness=1.00] (3.center) to (6);
		\draw [style=diredge] (6) to (7);
		\draw [in=-90, out=90, looseness=0.75] (4.center) to (2.center);
		\draw [in=-90, out=90, looseness=0.75] (10.center) to (3.center);
	\end{pgfonlayer}
\end{tikzpicture}
  \end{equation}
  and \emph{special} when
  \[\begin{tikzpicture}[dotpic]
	\begin{pgfonlayer}{nodelayer}
		\node [style=none] (0) at (1, 0) {$=$};
		\node [style=none] (1) at (-1, -1.75) {};
		\node [style=none] (2) at (-1, 1.75) {};
		\node [style=none] (3) at (-1.75, 0) {};
		\node [style=none] (4) at (-0.25, 0) {};
		\node [style=none] (5) at (-1.75, 0) {};
		\node [style=none] (6) at (-0.25, 0) {};
		\node [style=white dot] (7) at (-1, -1) {};
		\node [style=none] (8) at (-1.75, 0) {};
		\node [style=none] (9) at (-0.25, 0) {};
		\node [style=white dot] (10) at (-1, 1) {};
		\node [style=none] (11) at (-1.75, 0) {};
		\node [style=none] (12) at (-0.25, 0) {};
		\node [style=none] (13) at (2.25, -1.75) {};
		\node [style=none] (14) at (2.25, 1.75) {};
	\end{pgfonlayer}
	\begin{pgfonlayer}{edgelayer}
		\draw [style=swap, in=15, out=-90, looseness=1.00] (4.center) to (7);
		\draw [style=swap, in=-90, out=165, looseness=1.00] (7) to (3.center);
		\draw [style=diredge, in=-15, out=90, looseness=1.00] (9.center) to (10);
		\draw [style=diredge, in=-165, out=90, looseness=1.00] (11.center) to (10);
		\draw [style=diredge] (13.center) to (14.center);
		\draw [style=diredge] (10) to (2.center);
		\draw [style=diredge] (1.center) to (7);
	\end{pgfonlayer}
\end{tikzpicture}.\]
  We call a dagger Frobenius structure that is symmetric and special an \textit{abstract C*-algebra}.
\end{definition}


It was shown in \cite{Vicary2011} that dagger Frobenius structures in \FHilb correspond exactly to finite-dimensional C*-algebras, and in \cite{heunenkissingerselinger:idempotents}, it was furthermore shown that any such C*-algebra is isomorphic to one that is special and symmetric.

\begin{example}
  For a $D$-dimensional Hilbert space $H$, any orthonormal basis $\{\, \ket{\phi_i} \, \}_i$ defines a commutative abstract C*-algebra on $H$ via the \emph{Schur product}. This can be expressed in Dirac `bra-ket' notation as follows:
  \[
  \whitemult := \sum_i \ket{\phi_i} (\bra{\phi_i} \otimes \bra{\phi_i})
  \qquad\quad
  \whiteunit := \sum_i \ket{\phi_i}
  \]
  The comultiplication and counit are given as the adjoints of these two maps, namely:
  \[
  \whitecomult := \sum_i (\ket{\phi_i} \otimes \ket{\phi_i})\bra{\phi_i}
  \qquad\quad
  \whitecounit := \sum_i \bra{\phi_i}
  \]
\end{example}

\begin{example}
The algebra $\mathcal B(H)$ is almost an abstract C*-algebra, but it fails to satisfy the specialness equation by a scaling factor. However by normalising:
  \[
  \whitemult(M \otimes N) = \oneoversqrtD\, \cdot \, M N
  \qquad\quad
  \whiteunit = \sqrtD\, \cdot \, 1
  \]
  where $D$ is the dimension of $H$, we obtain an abstract C*-algebra $\mathcal A$ isomorphic to $\mathcal B(H)$. The comultiplication and counit are formed by taking the adjoint with respect to the Hilbert-Schmidt inner product on $\mathcal B(H)$, giving:
  \[
  \whitecomult(e_{ij}) = \oneoversqrtD\, \cdot \sum_k e_{ik} \otimes e_{kj}
  \qquad\quad
  \whitecounit(M) = \sqrtD\, \cdot \, \tr(M)
  \]
  where $e_{ij} = \ket{\phi_i}\bra{\phi_j}$ for any choice of orthonormal basis of $H$.
  Thus, symmetry~\eqref{eq:symmetric} captures cyclicity of the trace:
  \[ \tr(M N) = \tr(N M) \]
  For any abstract C*-algebra in \FHilb, we define a `star' on elements of the algebra via:
  \begin{equation}\label{eq:star}
    \begin{tikzpicture}[dotpic]
	\begin{pgfonlayer}{nodelayer}
		\node [style=square box] (0) at (-2, -0.5) {$a$};
		\node [style=none] (1) at (-2, 1) {};
		\node [style=none] (2) at (-0.25, 0) {$\mapsto$};
		\node [style=white dot] (3) at (3, 0.5) {};
		\node [style=square box] (4) at (4.25, -1) {$a$};
		\node [style=white dot] (5) at (3, 1.5) {};
		\node [style=none] (6) at (2, -0.75) {};
		\node [style=none] (7) at (2, -1.75) {};
		\node [style=none] (8) at (3.5, 0) {$\left(\phantombox{(1.6,1.8)}\right)^\dagger$};
	\end{pgfonlayer}
	\begin{pgfonlayer}{edgelayer}
		\draw [style=diredge] (0) to (1.center);
		\draw [style=diredge, in=-15, out=90, looseness=1.00] (4) to (3);
		\draw [style=diredge, in=-150, out=90, looseness=1.00] (6.center) to (3);
		\draw [style=diredge] (3) to (5);
		\draw (7.center) to (6.center);
	\end{pgfonlayer}
\end{tikzpicture}
  \end{equation}
  Then, the C*-algebraic norm is fixed uniquely as the spectral radius by the spectral theorem. For more details see \cite{Vicary2011}.
\end{example}

While abstract C*-algebras correspond exactly to C*-algebras in \FHilb, they correspond to a familiar, but very different kind of structure in $\cat{Rel}$, namely \textit{groupoids}.

\begin{definition}
A \textit{groupoid} is a category whose morphisms are all isomorphisms.
\end{definition}

Groupoids generalise groups, in the sense that a group is the same thing as a groupoid with a single object. Groupoids are a useful tool for capturing symmetries or other geometric information may not be fully captured by a group. For example, in homotopy theory, one can define the fundamental groupoid of a topological space without fixing a basepoint~\cite{brown2006topology}, which is sometimes more convenient than the fundamental group, e.g. for disconnected spaces.

\begin{example}
  In the category $\cat{Rel}$ of sets and relations, special dagger Frobenius structures correspond on an object $A$ precisely to groupoids whose set of morphisms is $A$. The multiplication is the relation
  \[
    \whitemult = \{((g,f),g \circ f) \mid f,g \in A, g \circ f \emph{ is defined}\}
  \] 
  and the unit is the set of identities:
  \[
    \whiteunit = \{ (*,\id[X]) \mid X \text{ is an object in the groupoid} \} \subseteq I \times A
  \] 
  The comultiplication and counit are simply the relational converses of the multiplication and unit. Finally, just like for C*-algebras in $\cat{FHilb}$, there is a `star' operation on elements of an abstract C*-algebra in $\cat{Rel}$ given by~\eqref{eq:star}, which relates $g$ to its inverse $g^{-1}$.
  For more details see~\cite{heunencontrerascattaneo:groupoids}.
\end{example}

A useful calculational tool for symmetric Frobenius structures (and hence abstract C*-algebras) is the \emph{symmetric spider theorem}. It tells us which of a wide variety of diagrams involving the algebraic structure are equal. Call a diagram a \emph{tree} when its underlying (undirected) graph is a tree (\textit{i.e.}\ is connected and acyclic) and \emph{planar} if it contains no wire-crossings.

\begin{theorem}[Symmetric spider]
  For a symmetric dagger Frobenius structure $(A, \whitemult, \whiteunit)$, suppose $f$ and $g$ can be written as planar trees consisting only of $\whitemult$, $\whiteunit$, caps, cups, and $1_A$ with the same domain and codomain. Then $f = g$.
\end{theorem}

\begin{proof}
  See \cite{QuickThesis}.
\end{proof}

Since any planar tree with a given input/output type is equivalent, we can collapse them into a single node without ambiguity:
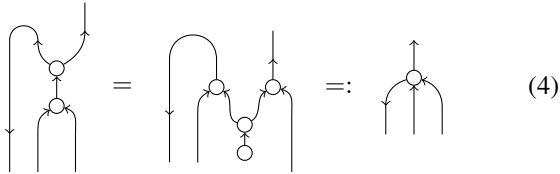
\begin{equation}\label{eq:symm-spider}
  \begin{tikzpicture}[dotpic]
	\begin{pgfonlayer}{nodelayer}
		\node [style=white dot] (0) at (-1.75, -0.5) {};
		\node [style=white dot] (1) at (-1.75, 0.5) {};
		\node [style=none] (2) at (-2.25, -1.25) {};
		\node [style=none] (3) at (-1.25, -1.25) {};
		\node [style=none] (4) at (-2.25, 1.25) {};
		\node [style=none] (5) at (-3, 1.25) {};
		\node [style=none] (6) at (-1, 1.5) {};
		\node [style=none] (7) at (-3, -1.25) {};
		\node [style=none] (8) at (-1, 2.25) {};
		\node [style=none] (9) at (-3, -2.25) {};
		\node [style=none] (10) at (-2.25, -2.25) {};
		\node [style=none] (11) at (-1.25, -2.25) {};
		\node [style=white dot] (12) at (2.5, 0) {};
		\node [style=none] (13) at (2, -0.75) {};
		\node [style=none] (14) at (2.5, 0.75) {};
		\node [style=none] (15) at (1.25, -0.75) {};
		\node [style=none] (16) at (1.25, 0.75) {};
		\node [style=white dot] (17) at (3.25, -1) {};
		\node [style=none] (18) at (2, -2) {};
		\node [style=none] (19) at (1.25, -2) {};
		\node [style=white dot] (20) at (3.25, -1.75) {};
		\node [style=white dot] (21) at (4, 0) {};
		\node [style=none] (22) at (4.5, -0.75) {};
		\node [style=none] (23) at (4, 0.75) {};
		\node [style=none] (24) at (4, 1.5) {};
		\node [style=none] (25) at (4.5, -2.25) {};
		\node [style=none] (26) at (0, 0) {$=$};
		\node [style=none] (27) at (8.5, -1.25) {};
		\node [style=none] (28) at (8.5, -0.5) {};
		\node [style=none] (29) at (7, -1.25) {};
		\node [style=none] (30) at (7.75, 1.25) {};
		\node [style=white dot] (31) at (7.75, 0.25) {};
		\node [style=none] (32) at (7, -0.5) {};
		\node [style=none] (33) at (5.75, 0) {$=:$};
		\node [style=none] (34) at (7.75, -1.25) {};
	\end{pgfonlayer}
	\begin{pgfonlayer}{edgelayer}
		\draw [style=diredge] (0) to (1);
		\draw [style=diredge, in=-150, out=90, looseness=1.25] (2.center) to (0);
		\draw [style=diredge, in=-15, out=90, looseness=1.25] (3.center) to (0);
		\draw [style=diredge, in=-90, out=150, looseness=1.00] (1) to (4.center);
		\draw [style=diredge] (5.center) to (7.center);
		\draw [style=diredge, in=-90, out=30, looseness=1.00] (1) to (6.center);
		\draw [in=90, out=90, looseness=1.75] (4.center) to (5.center);
		\draw (6.center) to (8.center);
		\draw [in=-90, out=90, looseness=1.00] (10.center) to (2.center);
		\draw [in=-90, out=90, looseness=1.00] (9.center) to (7.center);
		\draw (11.center) to (3.center);
		\draw [style=diredge, in=-150, out=90, looseness=1.25] (13.center) to (12);
		\draw [style=diredge] (16.center) to (15.center);
		\draw [in=90, out=90, looseness=1.75] (14.center) to (16.center);
		\draw (12) to (14.center);
		\draw [in=90, out=-90, looseness=1.00] (15.center) to (19.center);
		\draw [in=-90, out=90, looseness=1.00] (18.center) to (13.center);
		\draw [style=diredge] (20) to (17);
		\draw [style=diredge, in=-150, out=30, looseness=1.00] (17) to (21);
		\draw [style=diredge, in=-30, out=165, looseness=1.00] (17) to (12);
		\draw [style=diredge, in=-15, out=90, looseness=1.25] (22.center) to (21);
		\draw [style=diredge] (21) to (23.center);
		\draw (23.center) to (24.center);
		\draw (25.center) to (22.center);
		\draw [style=diredge] (31) to (30.center);
		\draw [style=diredge, in=90, out=-165, looseness=1.00] (31) to (32.center);
		\draw [style=diredge, in=-15, out=90, looseness=1.00] (28.center) to (31);
		\draw (29.center) to (32.center);
		\draw (27.center) to (28.center);
		\draw [style=diredge] (34.center) to (31);
	\end{pgfonlayer}
\end{tikzpicture}
\end{equation}
Such a node is called a \emph{spider}. This generalises the spider theorem for commutative Frobenius structures, which has been widely used in the categorical quantum mechanics literature~\cite{CD2,CQMII}. However, some care needs to be taken in the symmetric case, since the restriction to planarity means that the order of inputs and outputs is relevant. So, unlike for commutative Frobenius structures:
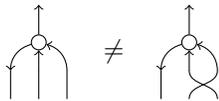
\ctikzfig{symm-neq}

Before showing the construction of a category whose objects are all abstract C*-algebras in $\cat C$, we first focus just on the abstract analogues to C*-algebras of the form $\mathcal B(H)$, for a finite-dimensional Hilbert space $H$, as was done in~\cite{SelingerCPM}. First, we note that morphisms of the form $\rho' \colon I \to H^* \otimes H$ are in 1-to-1 correspondence with morphisms $\rho\colon H \to H$. That is, there is a canonical bijection $\hom(H, H) \cong \hom(I, H^* \otimes H)$ given by `bending the wire', \textit{i.e.}\ pre-composing $\eta$:
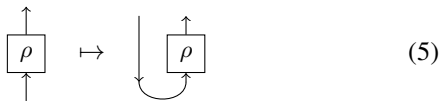
\begin{equation}\label{eq:bend-the-wire}
  \begin{tikzpicture}[dotpic]
	\begin{pgfonlayer}{nodelayer}
		\node [style=square box] (0) at (-1.75, 0) {$\rho$};
		\node [style=none] (1) at (-1.75, -1.25) {};
		\node [style=none] (2) at (-1.75, 1.25) {};
		\node [style=none] (3) at (0, 0) {$\mapsto$};
		\node [style=square box] (4) at (2.5, 0) {$\rho$};
		\node [style=none] (5) at (2.5, 1) {};
		\node [style=none] (6) at (2.5, -0.75) {};
		\node [style=none] (7) at (1.25, -0.75) {};
		\node [style=none] (8) at (1.25, -0.75) {};
		\node [style=none] (9) at (1.25, 1) {};
	\end{pgfonlayer}
	\begin{pgfonlayer}{edgelayer}
		\draw [style=diredge] (1.center) to (0);
		\draw [style=diredge] (0) to (2.center);
		\draw [style=diredge] (6.center) to (4);
		\draw [style=diredge] (4) to (5.center);
		\draw (8.center) to (7.center);
		\draw [in=-90, out=-90, looseness=1.25] (7.center) to (6.center);
		\draw [style=diredge] (9.center) to (8.center);
	\end{pgfonlayer}
\end{tikzpicture}
\end{equation}
whose inverse is given by post-composing by $\epsilon$
\begin{equation}\label{eq:unbend-the-wire}
  \begin{tikzpicture}[dotpic]
	\begin{pgfonlayer}{nodelayer}
		\node [style=square box, minimum width=8mm] (0) at (-2, -0.75) {$\psi$};
		\node [style=none] (1) at (-2.5, 0.25) {};
		\node [style=none] (2) at (-1.5, -0.25) {};
		\node [style=none] (3) at (-2.5, 1) {};
		\node [style=none] (4) at (-1.5, 0.5) {};
		\node [style=none] (5) at (-2.5, -0.25) {};
		\node [style=none] (6) at (-1.5, 1) {};
		\node [style=none] (7) at (0, 0) {$\mapsto$};
		\node [style=none] (9) at (1.5, 0.75) {};
		\node [style=none] (10) at (1.5, 0.75) {};
		\node [style=none] (11) at (1.5, -1.25) {};
		\node [style=square box, minimum width=8mm] (12) at (3.25, -0.75) {$\psi$};
		\node [style=none] (13) at (2.75, 0.5) {};
		\node [style=none] (14) at (3.75, -0.25) {};
		\node [style=none] (15) at (2.75, 0.75) {};
		\node [style=none] (16) at (3.75, 0.5) {};
		\node [style=none] (17) at (2.75, -0.25) {};
		\node [style=none] (18) at (3.75, 1.25) {};
	\end{pgfonlayer}
	\begin{pgfonlayer}{edgelayer}
		\draw [style=diredge] (3.center) to (1.center);
		\draw [style=diredge] (2.center) to (4.center);
		\draw (1.center) to (5.center);
		\draw (4.center) to (6.center);
		\draw (10.center) to (9.center);
		\draw [style=diredge] (11.center) to (10.center);
		\draw [style=diredge] (15.center) to (13.center);
		\draw [style=diredge] (14.center) to (16.center);
		\draw (13.center) to (17.center);
		\draw (16.center) to (18.center);
		\draw [in=90, out=90, looseness=1.25] (9.center) to (15.center);
	\end{pgfonlayer}
\end{tikzpicture}
\end{equation}

Sometimes the right-hand side of~\eqref{eq:bend-the-wire} is referred to as the \emph{name} of the morphism $\rho$. Using it, we can represent superoperators $\hom(H, H) \to \hom(K, K)$ as morphisms of type
\[ \Phi \colon H^* \otimes H \to K^* \otimes K\text{.} \]
Then, \cite{SelingerCPM} showed that, when $\cat C = \FHilb$, such morphisms correspond to CP-maps precisely when they factor as
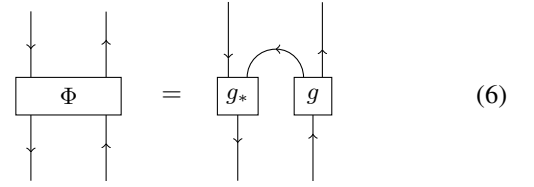
\begin{equation}\label{eq:cpm-cond}
  \begin{tikzpicture}[dotpic]
	\begin{pgfonlayer}{nodelayer}
		\node [style=square box] (0) at (4, 0) {$g$};
		\node [style=square box] (1) at (2, 0) {$g_*$};
		\node [style=none] (2) at (3.75, 0.5) {};
		\node [style=none] (3) at (4, -2.25) {};
		\node [style=none] (4) at (2.25, 0.5) {};
		\node [style=none] (5) at (3, 1.25) {};
		\node [style=none] (6) at (4, -1) {};
		\node [style=none] (7) at (0.25, 0) {$=$};
		\node [style=none] (8) at (4, -0.5) {};
		\node [style=none] (9) at (4.25, 1.25) {};
		\node [style=none] (10) at (4.25, 2.5) {};
		\node [style=none] (11) at (4.25, 0.5) {};
		\node [style=none] (12) at (1.75, 1) {};
		\node [style=none] (13) at (2, -2.25) {};
		\node [style=none] (14) at (2, -1.25) {};
		\node [style=none] (15) at (1.75, 0.5) {};
		\node [style=none] (16) at (1.75, 2.5) {};
		\node [style=none] (17) at (2, -0.5) {};
		\node [style=square box, minimum width=1.4 cm] (18) at (-2.5, 0) {$\Phi$};
		\node [style=none] (19) at (-3.5, -1.5) {};
		\node [style=none] (20) at (-3.5, -2.25) {};
		\node [style=none] (21) at (-3.5, -0.5) {};
		\node [style=none] (22) at (-3.5, 1.25) {};
		\node [style=none] (23) at (-3.5, 0.5) {};
		\node [style=none] (24) at (-3.5, 2.25) {};
		\node [style=none] (25) at (-1.5, 0.5) {};
		\node [style=none] (26) at (-1.5, -0.5) {};
		\node [style=none] (27) at (-1.5, -1.25) {};
		\node [style=none] (28) at (-1.5, 1.5) {};
		\node [style=none] (29) at (-1.5, 2.25) {};
		\node [style=none] (30) at (-1.5, -2.25) {};
	\end{pgfonlayer}
	\begin{pgfonlayer}{edgelayer}
		\draw [style=diredge, in=0, out=90, looseness=1.00] (2.center) to (5.center);
		\draw [style=diredge] (3.center) to (6.center);
		\draw [in=90, out=180, looseness=1.00] (5.center) to (4.center);
		\draw (6.center) to (8.center);
		\draw [style=diredge] (11.center) to (9.center);
		\draw (9.center) to (10.center);
		\draw [style=diredge] (16.center) to (12.center);
		\draw (12.center) to (15.center);
		\draw [style=diredge] (17.center) to (14.center);
		\draw (14.center) to (13.center);
		\draw [style=diredge] (21.center) to (19.center);
		\draw (19.center) to (20.center);
		\draw [style=diredge] (24.center) to (22.center);
		\draw (22.center) to (23.center);
		\draw [style=diredge] (25.center) to (28.center);
		\draw (28.center) to (29.center);
		\draw [style=diredge] (30.center) to (27.center);
		\draw (27.center) to (26.center);
	\end{pgfonlayer}
\end{tikzpicture}
\end{equation}
for some $g$. One way to see that this indeed gives the correct notion of CP-map in \FHilb is to start with the presentation of an arbitrary CP-map in terms of the partial trace:
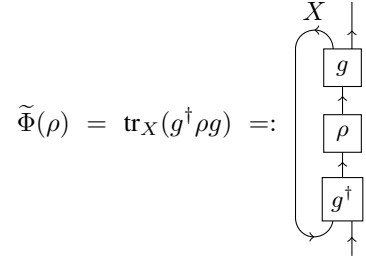
\[
\widetilde\Phi(\rho) \ =\  \tr_X(g^\dagger \rho g) \ =:\  \begin{tikzpicture}[dotpic]
	\begin{pgfonlayer}{nodelayer}
		\node [style=square box] (0) at (0, 1.5) {$g$};
		\node [style=none] (1) at (0.25, 2.5) {};
		\node [style=none] (2) at (-1.25, 2) {};
		\node [style=none] (3) at (-0.25, 2) {};
		\node [style=none] (4) at (-0.75, 2.5) {};
		\node [style=none] (5) at (0.25, 3.25) {};
		\node [style=none] (6) at (0.25, 2) {};
		\node [style=none] (7) at (-1.25, -2.5) {};
		\node [style=none] (8) at (0.25, -2.5) {};
		\node [style=none] (9) at (-0.25, -2.5) {};
		\node [style=square box] (10) at (0, -2) {$g^\dagger$};
		\node [style=none] (11) at (0.25, -3) {};
		\node [style=none] (12) at (-0.75, -3) {};
		\node [style=none] (13) at (0.25, -3.5) {};
		\node [style=square box] (14) at (0, -0.25) {$\rho$};
		\node [style=none] (15) at (0, -1.75) {};
		\node [style=none] (16) at (0, -1) {};
		\node [style=none] (17) at (0, -0.5) {};
		\node [style=none] (18) at (0, 0.75) {};
		\node [style=none] (19) at (0, 0) {};
		\node [style=none] (20) at (0, 1.25) {};
		\node [style=label] (21) at (-0.75, 3) {$X$};
	\end{pgfonlayer}
	\begin{pgfonlayer}{edgelayer}
		\draw [style=diredge, in=0, out=90, looseness=1.00] (3.center) to (4.center);
		\draw [in=90, out=180, looseness=1.00] (4.center) to (2.center);
		\draw [style=diredge] (6.center) to (1.center);
		\draw (1.center) to (5.center);
		\draw [in=-90, out=0, looseness=1.00] (12.center) to (9.center);
		\draw [style=diredge, in=180, out=-90, looseness=1.00] (7.center) to (12.center);
		\draw (11.center) to (8.center);
		\draw [style=diredge] (13.center) to (11.center);
		\draw (2.center) to (7.center);
		\draw [style=diredge] (15.center) to (16.center);
		\draw (16.center) to (17.center);
		\draw [style=diredge] (19.center) to (18.center);
		\draw (18.center) to (20.center);
	\end{pgfonlayer}
\end{tikzpicture}
\]
and `bend the wire' as in \eqref{eq:bend-the-wire}:
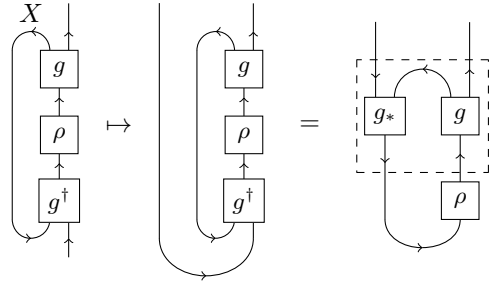
\[ \begin{tikzpicture}[dotpic]
	\begin{pgfonlayer}{nodelayer}
		\node [style=square box] (0) at (0, 1.5) {$g$};
		\node [style=none] (1) at (0.25, 2.5) {};
		\node [style=none] (2) at (-1.25, 2) {};
		\node [style=none] (3) at (-0.25, 2) {};
		\node [style=none] (4) at (-0.75, 2.5) {};
		\node [style=none] (5) at (0.25, 3.25) {};
		\node [style=none] (6) at (0.25, 2) {};
		\node [style=none] (7) at (-1.25, -2.5) {};
		\node [style=none] (8) at (0.25, -2.5) {};
		\node [style=none] (9) at (-0.25, -2.5) {};
		\node [style=square box] (10) at (0, -2) {$g^\dagger$};
		\node [style=none] (11) at (0.25, -3) {};
		\node [style=none] (12) at (-0.75, -3) {};
		\node [style=none] (13) at (0.25, -3.5) {};
		\node [style=square box] (14) at (0, -0.25) {$\rho$};
		\node [style=none] (15) at (0, -1.75) {};
		\node [style=none] (16) at (0, -1) {};
		\node [style=none] (17) at (0, -0.5) {};
		\node [style=none] (18) at (0, 0.75) {};
		\node [style=none] (19) at (0, 0) {};
		\node [style=none] (20) at (0, 1.25) {};
		\node [style=label] (21) at (-0.75, 3) {$X$};
	\end{pgfonlayer}
	\begin{pgfonlayer}{edgelayer}
		\draw [style=diredge, in=0, out=90, looseness=1.00] (3.center) to (4.center);
		\draw [in=90, out=180, looseness=1.00] (4.center) to (2.center);
		\draw [style=diredge] (6.center) to (1.center);
		\draw (1.center) to (5.center);
		\draw [in=-90, out=0, looseness=1.00] (12.center) to (9.center);
		\draw [style=diredge, in=180, out=-90, looseness=1.00] (7.center) to (12.center);
		\draw (11.center) to (8.center);
		\draw [style=diredge] (13.center) to (11.center);
		\draw (2.center) to (7.center);
		\draw [style=diredge] (15.center) to (16.center);
		\draw (16.center) to (17.center);
		\draw [style=diredge] (19.center) to (18.center);
		\draw (18.center) to (20.center);
	\end{pgfonlayer}
\end{tikzpicture} \ \ \mapsto\ \ \begin{tikzpicture}[dotpic]
	\begin{pgfonlayer}{nodelayer}
		\node [style=square box] (0) at (0, 1.5) {$g$};
		\node [style=none] (1) at (0.25, 2.5) {};
		\node [style=none] (2) at (-1.25, 2) {};
		\node [style=none] (3) at (-0.25, 2) {};
		\node [style=none] (4) at (-0.75, 2.5) {};
		\node [style=none] (5) at (0.25, 3.25) {};
		\node [style=none] (6) at (0.25, 2) {};
		\node [style=none] (7) at (-1.25, -2.5) {};
		\node [style=none] (8) at (0.25, -2.5) {};
		\node [style=none] (9) at (-0.25, -2.5) {};
		\node [style=square box] (10) at (0, -2) {$g^\dagger$};
		\node [style=none] (11) at (0.25, -3) {};
		\node [style=none] (12) at (-0.75, -3) {};
		\node [style=square box] (13) at (0, -0.25) {$\rho$};
		\node [style=none] (14) at (0, -1.75) {};
		\node [style=none] (15) at (0, -1) {};
		\node [style=none] (16) at (0, -0.5) {};
		\node [style=none] (17) at (0, 0.75) {};
		\node [style=none] (18) at (0, 0) {};
		\node [style=none] (19) at (0, 1.25) {};
		\node [style=none] (20) at (-1, -4) {};
		\node [style=none] (21) at (0.25, -3) {};
		\node [style=none] (22) at (-2.25, -3) {};
		\node [style=none] (23) at (-2.25, 3.25) {};
		\node [style=none] (24) at (1.75, 0) {$=$};
		\node [style=none] (25) at (3.75, -1) {};
		\node [style=none] (26) at (6, 2.75) {};
		\node [style=none] (27) at (5.5, 0.75) {};
		\node [style=none] (28) at (5.75, -0.75) {};
		\node [style=none] (29) at (5.75, -1.5) {};
		\node [style=none] (30) at (4.75, 1.5) {};
		\node [style=none] (31) at (3.75, -2.5) {};
		\node [style=none] (32) at (3.5, 2.75) {};
		\node [style=none] (33) at (3.75, -0.25) {};
		\node [style=square box] (34) at (3.75, 0.25) {$g_*$};
		\node [style=none] (35) at (6, 1.5) {};
		\node [style=none] (36) at (5.75, -0.25) {};
		\node [style=none] (37) at (3.5, 1.25) {};
		\node [style=square box] (38) at (5.75, 0.25) {$g$};
		\node [style=none] (39) at (4, 0.75) {};
		\node [style=none] (40) at (6, 0.75) {};
		\node [style=none] (41) at (3.5, 0.75) {};
		\node [style=square box] (42) at (5.75, -2) {$\rho$};
		\node [style=none] (43) at (4.75, -3.25) {};
		\node [style=none] (44) at (5.75, -2.5) {};
		\node [style=none] (45) at (3.75, -2.5) {};
		\node [style=none] (46) at (3, 1.75) {};
		\node [style=none] (47) at (6.5, 1.75) {};
		\node [style=none] (48) at (6.5, -1.25) {};
		\node [style=none] (49) at (3, -1.25) {};
	\end{pgfonlayer}
	\begin{pgfonlayer}{edgelayer}
		\draw [style=diredge, in=0, out=90, looseness=1.00] (3.center) to (4.center);
		\draw [in=90, out=180, looseness=1.00] (4.center) to (2.center);
		\draw [style=diredge] (6.center) to (1.center);
		\draw (1.center) to (5.center);
		\draw [in=-90, out=0, looseness=1.00] (12.center) to (9.center);
		\draw [style=diredge, in=180, out=-90, looseness=1.00] (7.center) to (12.center);
		\draw (11.center) to (8.center);
		\draw (2.center) to (7.center);
		\draw [style=diredge] (14.center) to (15.center);
		\draw (15.center) to (16.center);
		\draw [style=diredge] (18.center) to (17.center);
		\draw (17.center) to (19.center);
		\draw [in=-90, out=0, looseness=1.00] (20.center) to (21.center);
		\draw [style=diredge, in=180, out=-90, looseness=1.00] (22.center) to (20.center);
		\draw (23.center) to (22.center);
		\draw [style=diredge, in=0, out=90, looseness=1.00] (27.center) to (30.center);
		\draw [style=diredge] (29.center) to (28.center);
		\draw [in=90, out=180, looseness=1.00] (30.center) to (39.center);
		\draw (28.center) to (36.center);
		\draw [style=diredge] (40.center) to (35.center);
		\draw (35.center) to (26.center);
		\draw [style=diredge] (32.center) to (37.center);
		\draw (37.center) to (41.center);
		\draw [style=diredge] (33.center) to (25.center);
		\draw (25.center) to (31.center);
		\draw [in=-90, out=0, looseness=1.00] (43.center) to (44.center);
		\draw [style=diredge, in=180, out=-90, looseness=1.00] (45.center) to (43.center);
		\draw [style=dashed edge] (48.center) to (47.center);
		\draw [style=dashed edge] (47.center) to (46.center);
		\draw [style=dashed edge] (46.center) to (49.center);
		\draw [style=dashed edge] (49.center) to (48.center);
	\end{pgfonlayer}
\end{tikzpicture} \]
Hence, $\Phi$ depicted in \eqref{eq:cpm-cond} sends the name of $\rho$ to the name of $\widetilde\Phi(\rho)$ for an arbitrary CP-map $\widetilde\Phi$.

To pass from CP-maps between $\mathcal B(H)$-type algebras to CP-maps between arbitrary abstract C*-algebras, we can embed $\hom(A, B)$ into $\hom(A^* \otimes A, B^* \otimes B)$ as follows:
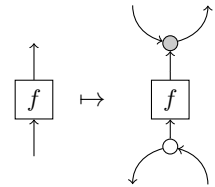
\[ \begin{tikzpicture}[dotpic]
	\begin{pgfonlayer}{nodelayer}
		\node [style=square box] (0) at (0, 0) {$f$};
		\node [style=none] (1) at (0, -1.5) {};
		\node [style=none] (2) at (0, 1.5) {};
	\end{pgfonlayer}
	\begin{pgfonlayer}{edgelayer}
		\draw [style=diredge] (0) to (2.center);
		\draw [style=diredge] (1.center) to (0);
	\end{pgfonlayer}
\end{tikzpicture}\ \ \mapsto\ \ \begin{tikzpicture}[dotpic]
	\begin{pgfonlayer}{nodelayer}
		\node [style=none] (0) at (0, -1) {};
		\node [style=white dot] (1) at (0, -1.25) {};
		\node [style=none] (2) at (1, -2.25) {};
		\node [style=none] (3) at (-1, -2.25) {};
		\node [style=square box] (4) at (0, 0) {$f$};
		\node [style=none] (5) at (-1, 2.5) {};
		\node [style=gray dot] (6) at (0, 1.5) {};
		\node [style=none] (7) at (1, 2.5) {};
		\node [style=none] (8) at (0, 0.75) {};
	\end{pgfonlayer}
	\begin{pgfonlayer}{edgelayer}
		\draw [in=90, out=-90, looseness=1.00] (0.center) to (1);
		\draw [style=diredge, in=90, out=-165, looseness=1.00] (1) to (3.center);
		\draw [style=diredge, in=-15, out=90, looseness=1.00] (2.center) to (1);
		\draw [style=diredge] (0.center) to (4);
		\draw [style=diredge, in=-90, out=90, looseness=1.00] (8.center) to (6);
		\draw [style=diredge, in=165, out=-90, looseness=1.00] (5.center) to (6);
		\draw [style=diredge, in=-90, out=15, looseness=1.00] (6) to (7.center);
		\draw (4) to (8.center);
	\end{pgfonlayer}
\end{tikzpicture} \]
We obtain `abstract CP-maps' by requiring that the image of this embedding satisfies~\eqref{eq:cpm-cond}. From this we get our main definition.

\begin{definition}
  For a compact dagger category $\cat{C}$, the category $\CPs[\cat{C}]$ has as objects special symmetric dagger Frobenius structures $(A, \whitemult)$, and as morphisms
  \[ (A, \whitemult) \to (B, \graymult) \]
  morphisms $f \colon A \to B$ from $\cat{C}$ satisfying the \emph{CP*-condition}, namely there exists a morphism $g \colon A \to X \otimes B$ in $\mathcal C$ satisfying
  \begin{equation}\label{eq:cps-cond}
    \begin{tikzpicture}[dotpic]
	\begin{pgfonlayer}{nodelayer}
		\node [style=none] (0) at (-2.5, -1) {};
		\node [style=white dot] (1) at (-2.5, -1.25) {};
		\node [style=none] (2) at (-1.5, -2.25) {};
		\node [style=none] (3) at (-3.5, -2.25) {};
		\node [style=square box] (4) at (-2.5, 0) {$f$};
		\node [style=none] (5) at (-3.5, 2.5) {};
		\node [style=gray dot] (6) at (-2.5, 1.5) {};
		\node [style=none] (7) at (-1.5, 2.5) {};
		\node [style=none] (8) at (-2.5, 0.75) {};
		\node [style=square box] (9) at (4, 0) {$g$};
		\node [style=square box] (10) at (2, 0) {$g_*$};
		\node [style=none] (11) at (3.75, 0.5) {};
		\node [style=none] (12) at (4, -2.25) {};
		\node [style=none] (13) at (2.25, 0.5) {};
		\node [style=none] (14) at (3, 1.25) {};
		\node [style=none] (15) at (4, -1) {};
		\node [style=none] (16) at (-0.25, 0) {$=$};
		\node [style=none] (17) at (4, -0.5) {};
		\node [style=none] (18) at (4.25, 1.25) {};
		\node [style=none] (19) at (4.25, 2.5) {};
		\node [style=none] (20) at (4.25, 0.5) {};
		\node [style=none] (21) at (1.75, 1) {};
		\node [style=none] (22) at (2, -2.25) {};
		\node [style=none] (23) at (2, -1.25) {};
		\node [style=none] (24) at (1.75, 0.5) {};
		\node [style=none] (25) at (1.75, 2.5) {};
		\node [style=none] (26) at (2, -0.5) {};
	\end{pgfonlayer}
	\begin{pgfonlayer}{edgelayer}
		\draw [in=90, out=-90, looseness=1.00] (0.center) to (1);
		\draw [style=diredge, in=90, out=-165, looseness=1.00] (1) to (3.center);
		\draw [style=diredge, in=-15, out=90, looseness=1.00] (2.center) to (1);
		\draw [style=diredge] (0.center) to (4);
		\draw [style=diredge, in=-90, out=90, looseness=1.00] (8.center) to (6);
		\draw [style=diredge, in=165, out=-90, looseness=1.00] (5.center) to (6);
		\draw [style=diredge, in=-90, out=15, looseness=1.00] (6) to (7.center);
		\draw (4) to (8.center);
		\draw [style=diredge, in=0, out=90, looseness=1.00] (11.center) to (14.center);
		\draw [style=diredge] (12.center) to (15.center);
		\draw [in=90, out=180, looseness=1.00] (14.center) to (13.center);
		\draw (15.center) to (17.center);
		\draw [style=diredge] (20.center) to (18.center);
		\draw (18.center) to (19.center);
		\draw [style=diredge] (25.center) to (21.center);
		\draw (21.center) to (24.center);
		\draw [style=diredge] (26.center) to (23.center);
		\draw (23.center) to (22.center);
	\end{pgfonlayer}
\end{tikzpicture}.
  \end{equation}
\end{definition}

If $\cat{C}$ is a compact dagger category, then so is $\CPs[\cat{C}]$~\cite[Theorem~3.3]{coeckeheunenkissinger:cpstar}.

Our key examples are:
\begin{itemize}
  \item The category $\CPs[\FHilb]$, which is equivalent to the category of finite-dimensional C*-algebras and completely positive linear maps~\cite[Proposition~3.5]{coeckeheunenkissinger:cpstar}.
  \item The category $\CPs[\Rel]$, which is equivalent to the category of groupoids and \emph{inverse-respecting} relations~\cite[Proposition~5.3]{coeckeheunenkissinger:cpstar}. The latter are relations $R \subseteq G \times H$ between the sets of morphisms of two groupoids satisfying:
    \begin{align}
      (g,h) \in R & \Longrightarrow (g^{-1},h^{-1}) \in R, \label{eq:homrel1} \\
      (g,h) \in R & \Longrightarrow (\id[\dom(g)],\id[\dom(h)]) \in R. \label{eq:homrel2}
    \end{align}
\end{itemize}

\begin{remark}
The CP*-condition comes in a variety of equivalent forms, coming from the fact that:
\begin{equation}\label{eq:act-coact}
  \begin{tikzpicture}[dotpic]
	\begin{pgfonlayer}{nodelayer}
		\node [style=none] (0) at (1, 0) {$=$};
		\node [style=none] (1) at (-1, -1.75) {};
		\node [style=none] (2) at (-1, 1.75) {};
		\node [style=none] (3) at (-1.75, 0) {};
		\node [style=none] (4) at (-0.25, 0) {};
		\node [style=none] (5) at (-1.75, 0) {};
		\node [style=none] (6) at (-0.25, 0) {};
		\node [style=white dot] (7) at (-1, -1) {};
		\node [style=none] (8) at (-1.75, 0) {};
		\node [style=none] (9) at (-0.25, 0) {};
		\node [style=white dot] (10) at (-1, 1) {};
		\node [style=none] (11) at (-1.75, 0) {};
		\node [style=none] (12) at (-0.25, 0) {};
		\node [style=none] (13) at (2.25, -1.75) {};
		\node [style=none] (14) at (2.25, 1.75) {};
	\end{pgfonlayer}
	\begin{pgfonlayer}{edgelayer}
		\draw [style=diredge, in=-90, out=15, looseness=1.00] (7) to (4.center);
		\draw [in=165, out=-90, looseness=1.00] (3.center) to (7);
		\draw [in=-15, out=90, looseness=1.00] (9.center) to (10);
		\draw [style=diredge, in=90, out=-165, looseness=1.00] (10) to (11.center);
		\draw [style=diredge] (13.center) to (14.center);
		\draw [style=diredge] (10) to (2.center);
		\draw [style=diredge] (1.center) to (7);
	\end{pgfonlayer}
\end{tikzpicture}
\end{equation}
(See also~\cite[Lemma~2.10]{coeckeheunenkissinger:cpstar}.) In particular, pre-composing both sides of \eqref{eq:cps-cond} with \whiteaction yields
\begin{equation}\label{eq:cps-equiv}
  \begin{tikzpicture}[dotpic]
	\begin{pgfonlayer}{nodelayer}
		\node [style=none] (0) at (-1.5, -1.75) {};
		\node [style=square box] (1) at (-1.5, -0.5) {$f$};
		\node [style=none] (2) at (-2.5, 2) {};
		\node [style=gray dot] (3) at (-1.5, 1) {};
		\node [style=none] (4) at (-0.5, 2) {};
		\node [style=none] (5) at (-1.5, 0.25) {};
		\node [style=square box] (6) at (4, 0.5) {$g$};
		\node [style=square box] (7) at (2, 0.5) {$g_*$};
		\node [style=none] (8) at (3.75, 1) {};
		\node [style=none] (9) at (4, -1) {};
		\node [style=none] (10) at (2.25, 1) {};
		\node [style=none] (11) at (3, 1.75) {};
		\node [style=none] (12) at (4, -0.5) {};
		\node [style=none] (13) at (0.25, 0) {$=$};
		\node [style=none] (14) at (4, 0) {};
		\node [style=none] (15) at (4.25, 1.75) {};
		\node [style=none] (16) at (4.25, 2.5) {};
		\node [style=none] (17) at (4.25, 1) {};
		\node [style=none] (18) at (1.75, 1.5) {};
		\node [style=none] (19) at (2, -1) {};
		\node [style=none] (20) at (2, -0.75) {};
		\node [style=none] (21) at (1.75, 1) {};
		\node [style=none] (22) at (1.75, 2.5) {};
		\node [style=none] (23) at (2, 0) {};
		\node [style=white dot] (24) at (3, -2) {};
		\node [style=none] (25) at (2, -1) {};
		\node [style=none] (26) at (3, -2.75) {};
		\node [style=none] (27) at (4, -1) {};
	\end{pgfonlayer}
	\begin{pgfonlayer}{edgelayer}
		\draw [style=diredge] (0.center) to (1);
		\draw [style=diredge, in=-90, out=90, looseness=1.00] (5.center) to (3);
		\draw [style=diredge, in=165, out=-90, looseness=1.00] (2.center) to (3);
		\draw [style=diredge, in=-90, out=15, looseness=1.00] (3) to (4.center);
		\draw (1) to (5.center);
		\draw [style=diredge, in=0, out=90, looseness=1.00] (8.center) to (11.center);
		\draw [style=diredge] (9.center) to (12.center);
		\draw [in=90, out=180, looseness=1.00] (11.center) to (10.center);
		\draw (12.center) to (14.center);
		\draw [style=diredge] (17.center) to (15.center);
		\draw (15.center) to (16.center);
		\draw [style=diredge] (22.center) to (18.center);
		\draw (18.center) to (21.center);
		\draw [style=diredge] (23.center) to (20.center);
		\draw (20.center) to (19.center);
		\draw [style=diredge, in=-90, out=90, looseness=1.00] (26.center) to (24);
		\draw [in=165, out=-90, looseness=1.00] (25.center) to (24);
		\draw [in=-90, out=15, looseness=1.00] (24) to (27.center);
	\end{pgfonlayer}
\end{tikzpicture}
\end{equation}
which has a familiar interpretation in \FHilb. Here, $(A, \whitemult)$ is a finite-dimensional C*-algebra whose underlying vector space $A$ happens to be a Hilbert space. Hence there $(A, \whitemult)$ is canonically represented on $\mathcal B(A)$ by right-multiplication. 
Letting $\pi$ and $\chi$ be the representations of $(A, \whitemult)$ on $\mathcal B(A)$ and $(B, \graymult)$ on $\mathcal B(B)$ respectively, condition \eqref{eq:cps-equiv} says that $f$ is a CP-map if and only if 
\[ \chi(f(a)) = \tr_X(g \pi(a) g^\dagger) \]
for all $a \in A$, which is essentially just the Stinespring dilation of $\chi \circ f$. For finite-dimensional C*-algebras, saying that $\chi \circ f$ has a Stinesping dilation is indeed equivalent to saying that $f$ is a CP-map~\cite{choi:completelypositive}.
Combining~\eqref{eq:act-coact} and~\eqref{eq:cps-equiv} shows that $f \colon A \to B$ satisfies the CP*-condition if and only if
\begin{equation}\label{eq:cpstarconvolution}
  \begin{tikzpicture}[dotpic]
	\begin{pgfonlayer}{nodelayer}
		\node [style=none] (0) at (2, -0.75) {};
		\node [style=square box, minimum width=1 cm, minimum height=0.75 cm] (1) at (4.5, 0) {$h$};
		\node [style=none] (2) at (0, 0) {$=$};
		\node [style=square box, minimum width=1 cm, minimum height=0.75 cm] (3) at (2, 0) {$h_*$};
		\node [style=none] (4) at (5, 0.75) {};
		\node [style=none] (5) at (4, 0.75) {};
		\node [style=none] (6) at (-2, 1.75) {};
		\node [style=none] (7) at (4.5, -0.75) {};
		\node [style=square box] (8) at (-2, 0) {$f$};
		\node [style=none] (9) at (1.5, 0.75) {};
		\node [style=none] (10) at (-2, -1.75) {};
		\node [style=none] (11) at (2.5, 0.75) {};
		\node [style=gray dot] (12) at (3.25, 2) {};
		\node [style=white dot] (13) at (3.25, -1.75) {};
		\node [style=none] (14) at (3.25, 2.75) {};
		\node [style=none] (15) at (3.25, -2.5) {};
	\end{pgfonlayer}
	\begin{pgfonlayer}{edgelayer}
		\draw [style=diredge] (8) to (6.center);
		\draw [style=diredge] (10.center) to (8);
		\draw [style=diredge, in=90, out=90, looseness=1.50] (5.center) to (11.center);
		\draw [style=diredge, in=0, out=90] (4.center) to (12);
		\draw [style=diredge, in=90, out=180] (12) to (9.center);
		\draw [style=diredge, in=180, out=-90] (0.center) to (13);
		\draw [style=diredge, in=-90, out=0] (13) to (7.center);
		\draw [style=diredge] (15.center) to (13);
		\draw [style=diredge] (12) to (14.center);
	\end{pgfonlayer}
\end{tikzpicture}
\end{equation}
for some object $X$ and morphism $h \colon A \to X \otimes B$.
\end{remark}

While we wish to remain as agnostic as possible about the base category $\cat C$, it is convenient to assume that objects $A$ in $\cat C$ are either a zero object $A \simeq 0$, or they are \emph{normalisable}, \textit{i.e.}\ there exists a positive isomorphism $\sqrt{d_A}\colon I \to I$ such that:
\[ \sqrt{d_A} \circ \sqrt{d_A} = \begin{tikzpicture}[dotpic]
	\begin{pgfonlayer}{nodelayer}
		\node [style=none] (0) at (0.5, 0) {};
		\node [style=none] (1) at (-1, 0) {};
		\node [style=right label] (2) at (0.75, 0) {$A$};
		\node [style=none] (3) at (-0.25, -0.75) {};
		\node [style=none] (4) at (-1, 0) {};
		\node [style=none] (5) at (0.5, 0) {};
		\node [style=none] (6) at (-0.25, 0.75) {};
	\end{pgfonlayer}
	\begin{pgfonlayer}{edgelayer}
		\draw [style=diredge, in=0, out=90, looseness=1.00] (0.center) to (6.center);
		\draw [in=90, out=180, looseness=1.00] (6.center) to (1.center);
		\draw [in=-90, out=0, looseness=1.00] (3.center) to (5.center);
		\draw [style=diredge, in=180, out=-90, looseness=1.00] (4.center) to (3.center);
	\end{pgfonlayer}
\end{tikzpicture} =: d_A \]
We refer to $d_A$ as the \textit{dimension} of $A$. In \FHilb, this indeed gives the dimension of the Hilbert space.

We now recall some basic facts about the CP*-construction from \cite{coeckeheunenkissinger:cpstar}. First, there is a functor $\mathcal{B} \colon \cat{C} \to \CPs[\cat{C}]$ which is analogous to the passage from a linear map $f : H \to K$ to a CP-map $\Psi_f \colon \mathcal B(H) \to \mathcal B(K)$ of the form:
\[ \Psi_f(\rho) = f \rho f^\dagger \]
Explicitly, the functor sends an object $A$ of $\cat{C}$ to the special symmetric dagger Frobenius structure $A^* \otimes A$ with the following multiplication and unit:
\[
  \left(\sqrt{d_A}\right)^{-1} \begin{tikzpicture}[dotpic]
	\begin{pgfonlayer}{nodelayer}
		\node [style=none] (0) at (-1, -1) {};
		\node [style=none] (1) at (0, 1) {};
		\node [style=none] (2) at (1, 1) {};
		\node [style=none] (3) at (2, -1) {};
		\node [style=none] (4) at (1, -1) {};
		\node [style=none] (5) at (0, -1) {};
	\end{pgfonlayer}
	\begin{pgfonlayer}{edgelayer}
		\draw [style=diredge, in=90, out=-90, looseness=1.00] (1.center) to (0.center);
		\draw [style=diredge, in=-90, out=90, looseness=1.00] (3.center) to (2.center);
		\draw [style=diredge, bend left=90, looseness=1.50] (5.center) to (4.center);
	\end{pgfonlayer}
\end{tikzpicture} \qquad \qquad \sqrt{d_A} \ \begin{tikzpicture}[dotpic]
	\path [use as bounding box] (-0.5,-0.25) rectangle (0.5,0.5);
	\begin{pgfonlayer}{nodelayer}
		\node [style=none] (0) at (0.5, 0.25) {};
		\node [style=none] (1) at (-0.5, 0.25) {};
		\node [style=none] (2) at (-0.5, 0.5) {};
		\node [style=none] (3) at (0.5, 0.5) {};
	\end{pgfonlayer}
	\begin{pgfonlayer}{edgelayer}
		\draw [bend right=90, looseness=1.75] (1.center) to (0.center);
		\draw (2.center) to (1.center);
		\draw [style=diredge] (0.center) to (3.center);
	\end{pgfonlayer}
\end{tikzpicture}
\]
The functor sends a morphism $f$ in $\cat{C}$ to $f_* \otimes f$, which is clearly of the form~\eqref{eq:cps-cond}.

Second, \emph{$*$-homomorphisms}, \textit{i.e.}\ morphisms in $\cat C$ satisfying:
\[\begin{tikzpicture}[dotpic,yscale=1.25,xscale=1.3]
	\begin{pgfonlayer}{nodelayer}
		\node [style=square box] (0) at (-5.75, -0.5) {$f$};
		\node [style=none] (1) at (-5.75, -1.5) {};
		\node [style=none] (2) at (-5.75, -0.5) {};
		\node [style=gray dot] (3) at (-5, 0.5) {};
		\node [style=square box] (4) at (-4.25, -0.5) {$f$};
		\node [style=none] (5) at (-4.25, -1.5) {};
		\node [style=none] (6) at (-5, 1.5) {};
		\node [style=none] (7) at (-2.75, 0) {$=$};
		\node [style=none] (8) at (-1.75, -1.5) {};
		\node [style=white dot] (9) at (-1, -0.5) {};
		\node [style=none] (10) at (-0.25, -1.5) {};
		\node [style=square box] (11) at (-1, 0.5) {$f$};
		\node [style=none] (12) at (-1, 1.5) {};
		\node [style=none] (13) at (2.5, 1.5) {};
		\node [style=none] (14) at (2.5, -1.5) {};
		\node [style=white dot] (15) at (2.5, -0.75) {};
		\node [style=square box] (16) at (2.5, 0.5) {$f$};
		\node [style=none] (17) at (3.5, 0) {$=$};
		\node [style=none] (18) at (4.5, 1.5) {};
		\node [style=none] (19) at (4.5, -1.5) {};
		\node [style=square box] (20) at (4.5, -0.5) {$f_*$};
		\node [style=gray dot] (21) at (4.5, 0.75) {};
	\end{pgfonlayer}
	\begin{pgfonlayer}{edgelayer}
		\draw [style=diredge, in=-150, out=90] (8.center) to (9);
		\draw [style=diredge] (9) to (11);
		\draw [style=diredge] (11) to (12.center);
		\draw [style=diredge] (15) to (16);
		\draw [style=diredge] (16) to (13.center);
		\draw [style=diredge] (15) to (14.center);
		\draw [style=diredge] (21) to (20);
		\draw [style=diredge] (20) to (19.center);
		\draw [style=diredge] (21) to (18.center);
		\draw [style=diredge] (1.center) to (0);
		\draw [style=diredge] (5.center) to (4);
		\draw [style=diredge, bend right] (4) to (3);
		\draw [style=diredge, bend left] (2.center) to (3);
		\draw [style=diredge] (3) to (6.center);
		\draw [style=diredge, bend right] (10.center) to (9);
	\end{pgfonlayer}
\end{tikzpicture},\]
where $\dotidualiser{white dot} = \whitecoaction \circ (\id \otimes \whiteunit) \colon A^* \to A$, satisfy the CP*-condition, and hence are also morphisms in $\CPs[\cat C]$. To see this, we first give an alternative characterisation of a $*$-homomorphism.

\begin{lemma}\label{lem:starhomo}
  Let $f \colon (A,\graymult) \to (C,\whitemult)$ be a $*$-homomorphism. Then:
  \ctikzfig{starhomodirected}
\end{lemma}
\begin{proof}
  First use that $f$ preserves involution, then that it preserves multiplication:
  \ctikzfig{starhomodirected_proof}  
  The second equations follows similarly.
\end{proof}

The fact that $*$-homomorphisms satisfy the CP*-condition now follows immediately from Lemma~\ref{lem:starhomo} and equation~\eqref{eq:act-coact}.

The final basic concept we need before exploring the characterisation of quantum theory is the notion of \textit{causality}, which tells us which processes are deterministically physically realisable.

\begin{definition}
  A morphism $f \colon (A,\whitemult) \to (B,\graymult)$ is called \textit{causal} if it preserves co-units:
  \begin{equation}\label{eq:causal}
  \begin{tikzpicture}[dotpic]
	\begin{pgfonlayer}{nodelayer}
		\node [style=square box] (0) at (-1.25, 0) {$f$};
		\node [style=none] (1) at (-1.25, 0.5) {};
		\node [style=gray dot] (2) at (-1.25, 1.25) {};
		\node [style=none] (3) at (-1.25, -1.25) {};
		\node [style=none] (4) at (0, 0) {$=$};
		\node [style=none] (5) at (1.25, -0.75) {};
		\node [style=white dot] (6) at (1.25, 0.25) {};
	\end{pgfonlayer}
	\begin{pgfonlayer}{edgelayer}
		\draw [style=diredge] (1.center) to (2);
		\draw [style=diredge] (5.center) to (6);
		\draw [style=diredge] (3.center) to (0);
	\end{pgfonlayer}
\end{tikzpicture}
  \end{equation}
\end{definition}

The causality condition says, intuitively, that if we discard the output of a process, then it doesn't matter which process happened. This condition originated in work in operational probabilistic theories~\cite{Chiri2}, but was later formulated for any symmetric monoidal category whose objects come with some notion of `discarding'~\cite{coeckelalcausal,coecke2016terminality}. While it may not seem immediately obvious from its form, this requirement efficiently captures the condition that the inputs of any morphism $f$ in a string diagram can only have an effect in the causal future of $f$. That is, they can only affect the outputs of $f$, the outputs of morphisms connected to the outputs of $f$, and so on.

The abstract relationship between the causality equation and causal influence in a string diagram is explained in detail in e.g.~\cite{PQP,heunenvicary:cqt}, but for our purposes, it suffices to consider the causality equation as a type of normalisation condition.
In the case of CP-maps $\mathcal B(H) \to \mathcal B(K)$,  satisfying equation~\eqref{eq:causal} corresponds to being trace-preserving. For more general C*-algebras, this plays a similar role to the requirement of a CP-map to be unital.

\begin{remark}
  Note that our notion of causality requires a CP-map to preserve the \textit{counit}, rather than unit of the C*-algebra, which is more common in the algebraic quantum information literature. This is because we work in the Schr\"odinger picture, where morphisms $f \colon (A,\whitemult) \to (B,\graymult)$ represent processes which take states of system $A$ to states of system $B$. Since all the categories we work with are dagger categories (and are therefore isomorphic to their opposite categories), this is equivalent to the Heisenberg picture where \textit{unital} CP-maps $f' \colon (B,\graymult) \to (A,\whitemult)$ represent processes from $A$ to $B$.
\end{remark}

\section{Signalling}\label{sec:signalling}

This section focuses on the relationship between signalling and kinematic independence. To be able to capture these concepts for general $\CPs[\cat{C}]$, we need to introduce the notion of subsystem.

\begin{definition}
  Let $(C,\whitemult)$ be an abstract C*-algebra. A \emph{subsystem} is another abstract C*-algebra $(A,\graymult)$ together with a morphism $f \colon A \to C$ satisfying $f^\dag \circ f = \id[A]$ that is a unital $*$-homomorphism. We call $f$ the \emph{inclusion} of the subsystem, and depict it as $\subsystem[gray]{A}{C}$.
\end{definition}

\begin{remark}
  For $f = \subsystem[gray]{A}{C} = \id[A] \otimes \blackunit$ we actually have that $f^\dagger \circ f = d_B \otimes 1_A$, rather than $1_A$. We can fix this by re-normalising to $(\sqrt{d_B})^{-1} \otimes f$. For simplicity, we will generally suppress such normalisation factors unless they are important.
\end{remark}

In general, if $(C,\whitemult) = (A \otimes B, \graymult\hspace*{-3mm}\blackmult)$, then $(A,\graymult)$ and $(B,\blackmult)$ are subsystems with inclusions $\subsystem[gray]{A}{C} = \id[A] \otimes \blackunit$ and $\subsystem[black]{B}{C} = \grayunit \otimes \id[B]$. But there can also be subsystems that are not tensor factors.
For $\cat{C}=\cat{FHilb}$, subsystems correspond precisely to C*-subalgebras. 
Next we look at subsystems for $\cat{C}=\cat{Rel}$.

\begin{example}\label{lem:subsystem:rel}
  In $\Rel$, a subsystem of a groupoid $\cat{G}$ consists of a groupoid $\cat{H}$ and a multi-valued function $R \colon H \to \mathcal{P}_{\neq\emptyset}(G)$ satisfying
  \begin{align}
    R(h^{-1}) & = R(h)^{-1}, \label{eq:inverses} \\
    R(h \circ h') & = R(h) \circ R(h'), \label{eq:multiplicative} \\
    \bigcup_{x \in \cat{H}} R(\id[x]) & = \bigcup_{y \in \cat{G}} \{\id[y]\}, \label{eq:unital} \\
    R(h) \cap R(h') = \emptyset & \Leftrightarrow h = h'. \label{eq:isometry}
  \end{align}
  We will call such subsystems \emph{wide subgroupoids}. This is a slight abuse of terminology: the image of $R$ is an honest wide subgroupoid, but $R$ itself may map one morphism of $\cat{H}$ to many morphisms of $\cat{G}$. 

  To see the above statement, let a groupoid $\cat{H}$ and relation $R \subseteq H \times G$ form a subsystem of $\cat{G}$.
  As any relation, we may regard $R$ as a function $R \colon H \to \mathcal{P}(G)$.
  Isometry then says that $R(h) \neq \emptyset$, and that $R(h) \cap R(h')=\emptyset$ when $h \neq h'$, that is~\eqref{eq:isometry}.
  Hence $R$ is a multi-valued function.
  In these terms, $R$ being a unital $*$-homomorphism translates into~\eqref{eq:inverses}--\eqref{eq:unital}.

  There is some more structure to subsystems. If $g,g' \in R(h)$ have the same codomain, then
  \begin{align*}
    g^{-1} \circ g'
    \in R(h)^{-1} \circ R(h)
    & = R(h^{-1} \circ h) \\
    & = R(\id[\dom(h)]) \\
    & \subseteq \{ \id[y] \mid y \in \cat{G} \},
  \end{align*}
  so $g=g'$. Thus $R$ is in fact a single-valued function when restricted to subsets of $H$ of morphisms with common codomain.
\end{example}

The we call the dagger of a subsystem inclusion a subsystem \textit{restriction}:
\[ \subsystemdag{}{}{}\ =\ \left(\subsystem{}{}{}\right)^\dagger \]
Since subsystem inclusions are unital, restrictions are always causal:
\ctikzfig{subsys-causal}
and hence physically realisable. Restrictions furthermore have a clear operational interpretation: they correspond to the process of discarding (or ignoring) everything which is not local to the subsystem $A$.

Subsystems also allow us to say when an operation on the large system is actually localised to a subsystem.

\begin{definition}\label{def:local} 
  An endomorphism of $(C,\whitemult)$ in $\CPs[\cat{C}]$ is called \emph{local} to a subsystem $(A,\graymult)$ when it is of the form:
  \begin{equation}\label{eq:local-to-A}
    \begin{tikzpicture}[dotpic]
	\begin{pgfonlayer}{nodelayer}
		\node [style=none] (0) at (-1, 0) {};
		\node [style=none] (1) at (1, 0) {};
		\node [style=white dot] (2) at (0, 1.75) {};
		\node [style=None] (3) at (0, 2.75) {};
		\node [style=None] (4) at (0, -2.75) {};
		\node [style=none] (5) at (-1, 1) {};
		\node [style=none] (6) at (1, 1) {};
		\node [style=square box] (7) at (-1, -0.5) {$E_*$};
		\node [style=square box] (8) at (1, -0.5) {$E$};
	\end{pgfonlayer}
	\begin{pgfonlayer}{edgelayer}
		\draw (4) to (2);
		\draw [style=diredge] (2) to (3);
		\draw [style=subsystem] (1.center) to (6.center);
		\draw [style=diredge, in=-30, out=90, looseness=0.75] (6.center) to (2);
		\draw [style=none, in=90, out=-150, looseness=1.00] (2) to (5.center);
		\draw [style=subsystem] (0.center) to (5.center);
		\draw [style=diredge, cross, in=-90, out=-90, looseness=1.50] (7) to (8);
	\end{pgfonlayer}
\end{tikzpicture}
  \end{equation}
\end{definition}

The above definition captures abstractly the C*-algebraic notion of a CP-map only having support on a sub-algebra, i.e.\ a CP-map $\Phi : C \to C$ being of the form $\Phi(x) = \sum_i e_i x e_i^*$ for $e_i$ all elements of a sub-algebra $A$ of $C$. We will now see that when these sub-systems arise from tensor factors, this recovers the usual notion of being localised to a tensor factor, i.e.\ being of the form $\Phi \otimes 1$.

\begin{proposition}
  Let $(A,\whitemult)$ and $(B,\graymult)$ be objects in $\CPs[\cat{C}]$. 
  If a morphism $A \otimes B \to A \otimes B$ in $\CPs[\cat{C}]$ is local to $A$, then it is of the form $g \otimes \id[B]$ for some $g \colon (A,\whitemult) \to (A,\whitemult)$. Conversely, if $(A,\whitemult) = (C^*\otimes C, \pantsalg)$, then every such map has the form \eqref{eq:local-to-A}.
\end{proposition}
\begin{proof}
  Unfolding Definition~\ref{def:local} shows that the morphism $A \otimes B \to A \otimes B$ has the form
  \ctikzfig{localtensor}
  for some morphisms $E$. Thus we may take the left half of the right-hand diagram for $g$. 
  For the converse, it suffices to show that any endomorphism of $(C^*\otimes C, \pantsalg)$ takes the desired form. That is, for any $h$, there exists $E$ such that:
  \ctikzfig{all-bh-interior}
  Taking
  \ctikzfig{E-h-flip}
  this immediately follows from diagram deformation.
\end{proof}

For the case of local maps, causality takes a simpler form.

\begin{lemma}\label{lem:local-caus}
  A map of the form \eqref{eq:local-to-A} is causal if and only if:
  \ctikzfig{local-caus}
\end{lemma}

\begin{proof}
  Causality is given by the following equation:
  \ctikzfig{local-caus-pf}
  Applying a spider with two outputs to both sides yields:
  \ctikzfig{local-caus-pf1}
  The right-hand side equals the unit, by the symmetric spider theorem. If we deform the left-hand diagram, we can also apply the symmetric spider theorem there, yielding:
  \ctikzfig{local-caus-pf2}
  This finishes the proof.
\end{proof}

We can now introduce the two concepts of interest. First, we consider kinematic independence. This concept formalises when systems controlled by Alice and Bob do not influence each other's kinematics, namely when the associated sub-algebras commute with respect to each other. The following definition captures this in such a way
that for $\cat{C}=\FHilb$ it coincides with the notion given in~\cite{cliftonbubhalvorson}.

\begin{definition}\label{def:kinematicindependence}
  Let $(C,\whitemult)$ be a dagger Frobenius structure in a compact dagger category $\cat{C}$. We say that two subsystems $(A,\graymult)$ and $(B,\blackmult)$, with inclusions $\subsystem{A}{C}$ and $\subsystem[black]{B}{C}$, are \emph{kinematically independent} when the following equation is satisfied. 
  \begin{equation}\label{eq:kinematicindepence}
    \begin{tikzpicture}[dotpic]
	\begin{pgfonlayer}{nodelayer}
		\node [style=none] (0) at (0, 0.75) {$=$};
		\node [style=None] (1) at (1, -1) {$A$};
		\node [style=None] (2) at (1.5, 2.25) {$C$};
		\node [style=None] (3) at (2, -1) {$B$};
		\node [style=white dot] (4) at (1.5, 1) {};
		\node [style=none] (5) at (1, 0) {};
		\node [style=none] (6) at (2, 0) {};
		\node [style=none] (7) at (-2, 0) {};
		\node [style=None] (8) at (-1, -1) {$B$};
		\node [style=None] (9) at (-2, -1) {$A$};
		\node [style=white dot] (10) at (-1.5, 1) {};
		\node [style=None] (11) at (-1.5, 2.25) {$C$};
		\node [style=none] (12) at (-1, 0) {};
	\end{pgfonlayer}
	\begin{pgfonlayer}{edgelayer}
		\draw [style=diredge] (4) to (2);
		\draw [style=subsystem] (1) to (5.center);
		\draw [style=subsystem, subsystemcolour=black] (3) to (6.center);
		\draw [style=diredge, in=-150, out=90, looseness=1.50] (6.center) to (4);
		\draw [style=diredge, in=-30, out=90, looseness=1.50] (5.center) to (4);
		\draw [style=diredge] (10) to (11);
		\draw [style=subsystem] (9) to (7.center);
		\draw [style=subsystem, subsystemcolour=black] (8) to (12.center);
		\draw [style=diredge, in=-45, out=90, looseness=1.25] (12.center) to (10);
		\draw [style=diredge, in=-135, out=90] (7.center) to (10);
	\end{pgfonlayer}
\end{tikzpicture}
  \end{equation}
\end{definition}

Note that, by applying the dagger to both sides, we can equivalently state kinematic indepences as:
\begin{equation}\label{eq:kinematicindepence-dag}
  \begin{tikzpicture}[dotpic]
	\begin{pgfonlayer}{nodelayer}
		\node [style=none] (0) at (0, 0.5) {$=$};
		\node [style=None] (1) at (1, 2.25) {$A$};
		\node [style=None] (2) at (1.5, -1) {$C$};
		\node [style=None] (3) at (2, 2.25) {$B$};
		\node [style=white dot] (4) at (1.5, 0.25) {};
		\node [style=none] (5) at (1, 1.25) {};
		\node [style=none] (6) at (2, 1.25) {};
		\node [style=none] (7) at (-2, 1.25) {};
		\node [style=None] (8) at (-1, 2.25) {$B$};
		\node [style=None] (9) at (-2, 2.25) {$A$};
		\node [style=white dot] (10) at (-1.5, 0.25) {};
		\node [style=None] (11) at (-1.5, -1) {$C$};
		\node [style=none] (12) at (-1, 1.25) {};
	\end{pgfonlayer}
	\begin{pgfonlayer}{edgelayer}
		\draw [style=diredge] (2) to (4);
		\draw [style=subsystem] (1) to (5.center);
		\draw [style=subsystem, subsystemcolour=black] (3) to (6.center);
		\draw [style=diredge, in=-90, out=150, looseness=1.50] (4) to (6.center);
		\draw [style=diredge, in=-90, out=30, looseness=1.50] (4) to (5.center);
		\draw [style=diredge] (11) to (10);
		\draw [style=subsystem] (9) to (7.center);
		\draw [style=subsystem, subsystemcolour=black] (8) to (12.center);
		\draw [style=diredge, in=-90, out=45, looseness=1.25] (10) to (12.center);
		\draw [style=diredge, in=-90, out=135, looseness=1.00] (10) to (7.center);
	\end{pgfonlayer}
\end{tikzpicture}
\end{equation}
This will be important when it comes to relating this concept to nonsignalling.

Notice that if $(C,\whitemult)$ decomposes into a tensor product $(A \otimes B, \graymult\hspace*{-3mm}\blackmult)$, then the subsystems $(A,\graymult)$ and $(B,\blackmult)$ are always kinematically independent. In other words, kinematic independence is a notion that essentially concerns subsystems that are not tensor factors. For example, in $\cat{C}=\cat{FHilb}$, kinematic independence means that the C*-subalgebras $A$ and $B$ commute. 

\begin{example}\label{ex:kinematicindependence:rel}
  Let $\cat{G}$ be a groupoid. Wide subgroupoids $\cat{A}$ and $\cat{B}$ are kinematically independent if and only if $S \circ T = T \circ S$ for all subsets $S \subseteq \cat{A}$ and $T \subseteq \cat{B}$ of morphisms. 
  Equivalently, when for all subsets $S \subseteq \cat{A}, T \subseteq \cat{B}$ and elements $a \in S$ and $b \in T$, if $a \circ b$ is defined then there exist $a' \in S, b' \in T$ such that $a \circ b = b' \circ a'$, and when $b \circ a$ is defined, there exist $a'' \in S, b'' \in T$ such that $a'' \circ b'' = b \circ a$.
\end{example}
\begin{proof}
  Plug the states $S$ and $T$ into Definition~\ref{def:kinematicindependence}. 
  Conversely, if $S \circ T = T \circ S$ for all states, then~\eqref{eq:kinematicindepence}.
\end{proof}

The next notion we consider prohibits superluminal information transfer. It says that when Alice and Bob both control a system, any data that Alice extracts from her system (through measurement) cannot instantaneously influence Bob's 
system. We formalise this as follows.

\begin{definition}\label{def:nosignalling}
  Let $(C,\whitemult)$ be a dagger Frobenius structure in a compact dagger category. Two subsystems $\subsystem{}{} \colon (A,\graymult) \to (C,\whitemult)$ and $\subsystem[black]{}{} \colon (B,\blackmult) \to (C,\whitemult)$ are \emph{no signalling} when
  \ctikzfig{nosignalling}
  for any $E, F$ which define causal maps local to $A$ and $B$, respectively.
\end{definition}

Again, tensor factors are automatically no signalling, making this notion essentially about subsystems that are not tensor factors: 
if $(C,\whitemult)$ decomposes into a tensor product $(A \otimes B, \graymult\hspace*{-3mm}\blackmult)$, then the subsystems $(A,\graymult)$ and $(B,\blackmult)$ are always no signalling.
For $\cat{C}=\FHilb$, our definition of no signalling comes down to the usual one employed in~\cite{cliftonbubhalvorson}. 

\begin{theorem}
  In $\CPs[\cat{C}]$ for compact dagger $\cat{C}$:  
  \[
    \text{no signalling} \;\Leftrightarrow\; \text{kinematic independence}
  \]
\end{theorem}
\begin{proof}
  First, assume kinematic indepence. Then applying the symmetric spider theorem and equation \eqref{eq:kinematicindepence-dag}:
  \ctikzfig{ki_implies_ns}
  A second application of the symmetric spider theorem, as well as Lemma~\ref{lem:local-caus} yields the first no-signalling equation:
  \ctikzfig{ki_implies_ns1}
  The second equation is similar.

  Conversely, assume no signalling. First note that taking $E = \id[A]$ in \eqref{eq:local-to-A} yields a causal map, by Lemma~\ref{lem:local-caus}:
  \ctikzfig{id-local-caus}
  Hence the first no signalling equation from Definition~\ref{def:nosignalling} applies:
  \ctikzfig{nosignalling-id}

  Applying this to the left-hand side of~\eqref{eq:kinematicindepence-dag} introduces a loop:
  \ctikzfig{nosig-pf}
  We then apply (the dagger of) the homomorphism equation:
  \ctikzfig{nosig-pf1}
  then apply the homomorphism on the other side:
  \ctikzfig{nosig-pf2}
  and finally remove the loop:
  \ctikzfig{nosig-pf3}
  This finishes the proof.
\end{proof}

\section{Broadcasting}\label{sec:broadcasting}

We now give a definition of broadcasting, generalising that of~\cite{cliftonbubhalvorson}.

\begin{definition}\label{def:CBHbroadcasting}
  Let $(C,\graymult)$ be an object of $\CPs[\cat{C}]$, and let
  $\subsystem[white]{}{}, \subsystem[black]{}{} \colon (A,\whitemult) \to (C,\graymult)$ be two kinematically independent subsystems of $(C,\graymult)$ whose domains are a fixed algebra $(A, \whitemult)$. A \emph{broadcasting map} is a morphism $B \colon A \to C$ in $\CPs[\cat{C}]$ satisfying the following equation.
  \begin{equation}\label{eq:CBHbroadcasting}
    \begin{tikzpicture}[scale=.66]
	\begin{pgfonlayer}{nodelayer}
		\node[style=square box] (0) at (0, 0) {$B$};
		\node[draw=none] at (1,0) {$=$};
		\node[draw=none] at (3,0) {$=$};
		\node[style=square box] (1) at (4, 0) {$B$};
	\end{pgfonlayer}
	\begin{pgfonlayer}{edgelayer}
		\draw [style=diredge] (2,-1.5) to (2,1.5);
		\draw [style=subsystem, subsystemcolour=white] ([yshift=1cm]0.north) to (0.north);
		\draw [style=diredge] ([yshift=-1cm]0.south) to (0.south);
		\draw [style=subsystem, subsystemcolour=black] ([yshift=1cm]1.north) to (1.north);
		\draw [style=diredge] ([yshift=-1cm]1.south) to (1.south);
	\end{pgfonlayer}
\end{tikzpicture}
  \end{equation}
  We say $A$ is \emph{broadcastable} when there exists a broadcasting map for some $C$.
\end{definition}

If $(A, \whitedot)$ is commutative, we can simply take $C := A$, and $\subsystem[white]{}{} = \subsystem[black]{}{} = \id[A]$. Hence, commutatitivity trivially implies broadcastability. In the case of $\CPs[\cat{FHilb}]$, broadcastability implies commutativity, thanks to~\cite[Theorem 3]{cliftonbubhalvorson}. However, as we will show in this section, when we pass to arbitrary $\CPs[\cat{C}]$, this is no longer the case.

As in the case of concrete C*-algebras, Definition~\ref{def:CBHbroadcasting} generalises the older, more familiar notion of broadcasting presented by Barnum et al~\cite{NoBroadcasting}, which we will call $\otimes$-broadcasting.

\begin{definition}\label{def:tensor-broadcasting}
  A \emph{$\otimes$-broadcasting map} for an object $(A,\whitemult)$ of $\CPs[\cat{C}]$ is a morphism $b \colon A \to A \otimes A$ in $\CPs[\cat{C}]$ satisfying the following equation.
  \begin{equation}\label{eq:broadcasting}
    \begin{tikzpicture}[dotpic]
	\begin{pgfonlayer}{nodelayer}
		\node [style=none] (0) at (-2.75, 0) {};
		\node [style=white dot] (1) at (-2.75, 1.25) {};
		\node [style=none] (2) at (-1.75, 1.5) {};
		\node [style=none] (3) at (-1.75, 0) {};
		\node [style=none] (4) at (-2.25, -1.5) {};
		\node [style=square box] (5) at (-2.25, 0) {$\;\;b\;\;$};
		\node [style=none] (6) at (-0.75, 0) {$=$};
		\node [style=none] (7) at (2, -1.5) {};
		\node [style=none] (8) at (1.5, 0) {};
		\node [style=none] (9) at (2.5, 0) {};
		\node [style=white dot] (10) at (2.5, 1.25) {};
		\node [style=none] (11) at (1.5, 1.5) {};
		\node [style=square box] (12) at (2, 0) {$\;\;b\;\;$};
		\node [style=none] (13) at (0.75, 0) {$=$};
		\node [style=none] (14) at (0, -1.5) {};
		\node [style=none] (15) at (0, 1.5) {};
	\end{pgfonlayer}
	\begin{pgfonlayer}{edgelayer}
		\draw [style=diredge] (4.center) to (5);
		\draw [style=diredge] (3.center) to (2.center);
		\draw [style=diredge] (0.center) to (1);
		\draw [style=diredge] (7.center) to (12);
		\draw [style=diredge] (8.center) to (11.center);
		\draw [style=diredge] (9.center) to (10);
		\draw [style=diredge] (14.center) to (15.center);
	\end{pgfonlayer}
\end{tikzpicture}
  \end{equation}
  The object $(A,\whitemult)$ is called \emph{$\otimes$-broadcastable} when
  there exists a $\otimes$-broadcasting map.
\end{definition}

\begin{lemma}\label{lem:commutativebroadcasting}
  Let $\cat{C}$ be a compact dagger category. Commutative dagger Frobenius structures in $\cat{C}$ are $\otimes$-broadcastable objects in $\CPs[\cat{C}]$.  
\end{lemma}
\begin{proof}
  Suppose that $(A,\whitemult)$ is commutative. We will show that
  $\whitecomult$ is a broadcasting map. It clearly
  satisfies~\eqref{eq:broadcasting}, so it suffices to show that it is
  a well-defined morphism in $\CPs[\cat{C}]$. Using the spider theorem for
  commutative dagger Frobenius structures~\cite[Lemma~3.1]{coeckepaquette:naimark}, we obtain the following:
  \ctikzfig{commutativebroadcasting}
  Therefore $\dotaction{white dot}\!\!\!\!\!\dotaction{white dot}
  \circ \whitecomult \circ \whitecoaction$ is a composition of the identity on $(A,\whitemult)$ in $\CPs[\V]$ and the image of $\whitecomult$ under the functor $\mathcal{B}$. Since these are both completely positive, so is their composition. Thus $\whitecomult$ is a well-defined morphism in $\CPs[\V]$. 
\end{proof}

In particular, we can conclude that commutative C*-algebras are not only broadcastable, but also $\otimes$-broadcastable.

A groupoid is \emph{skeletal} when its only morphisms are endomorphisms. Equivalently, it is a disjoint union of groups.

\begin{lemma}\label{lem:broadcastableRel}
  $\otimes$-broadcastable objects in $\CPs[\Rel]$ are precisely skeletal 
  groupoids.  
\end{lemma}
\begin{proof}
  Let $\cat{G}$ be a skeletal (small) groupoid, and write $G$ for its set of morphisms. We will show that the morphism
  $b \colon G \to G \times G$ in $\Rel$ given by
  \begin{align*}
    b = & \{ (f,(\id[\dom(f)],f)) \mid f \in \Mor(\cat{G}) \} \\
    & \cup \{ (f,(f,\id[\dom(f)])) \mid f \in \Mor(\cat{G}) \} 
  \end{align*}
  is a broadcasting map. First of all, $b$ is readily seen to respect identities (property \eqref{eq:homrel2}). Since $\cat{G}$ is skeletal, $\id[\dom(f)] = \id[\cod(f)] = \id[\dom(f^{-1})]$, so $f$ also preserves inverses (property \eqref{eq:homrel1}). Hence, it is a well-defined morphism in $\CPs[\Rel]$.

  When interpreted in $\Rel$, the broadcastability equation~\eqref{eq:broadcasting} reads:
  \begin{equation}
    \begin{aligned}
    \begin{array}{lcl}
    \{ (f,f) \mid f \in G \} &=& \{(f,g) \mid \exists C. (f,(\id[C],g)) \in
    b \} \\ &=& \{ (f,g) \mid \exists C. (f,(g,\id[C])) \in b \}.
    \end{array}
    \end{aligned}\tag{$*$}
  \end{equation}
  This is satisfied by $b$ as defined above. By definition, $(f, (\id[\dom(f)], f)$ is in $b$, so the LHS of the first equation above is contained in the RHS. On the other hand, if for some $g, C$, we have $(f, (\id[C], g)) \in b$, then $C = \dom(f)$ and $g = f$, so the LHS also contains the RHS. The second equation follows symmetrically.

  Conversely, suppose that a small groupoid $\cat{G}$ is broadcastable. Then there is a morphism $b$ in $\Rel$ that respects inverses, and satisfies ($*$). Let $f \in \Mor(\cat{G})$. By~($*$), there is an object $C$ of
  $\cat{G}$ such that $(f,(\id[C],f)) \in b$. Next, \eqref{eq:homrel2} gives $(\id[\dom(f)], (\id[C], \id[\dom(f)]))
  \in b$ and $C=\dom(f)$. But by~\eqref{eq:homrel1}, it is also the case that $(f^{-1},
  (\id[C],f^{-1})) \in B$. So, using~($*$) and~\eqref{eq:homrel2}
  again, we also have $(\id[\cod(f)], (\id[C], \id[\cod(f)]))$ and $C =
  \cod(f)$. Hence $\dom(f)=\cod(f)$. Thus $\cat{G}$ is skeletal.
\end{proof}

\begin{lemma}\label{lem:BBLWimpliesCBHbroadcasting}
  Let $b$ be an $\otimes$-broadcasting map for an object $(A,\whitemult)$. Then it is also a broadcasting map, where $C:=A \otimes A$, $\subsystem[white]{}{} = \idpic \whiteunit$, and $\subsystem[black]{}{}=\whiteunit \idpic$. Hence $\otimes$-broadcastability implies broadcastability.
\end{lemma}
\begin{proof}
  Simply unfold the definitions.  
\end{proof}

\begin{theorem}\label{thm:broadcastingvscommutativity}
  In $\CPs[\cat{C}]$ for general $\cat{C}$:
  \[
  \text{no broadcasting}
  \begin{aligned}\begin{array}{c}
    \Rightarrow \\[-1ex]
    \nLeftarrow
  \end{array}\end{aligned}
  \text{noncommutativity}
  \]
\end{theorem}
\begin{proof}
  The implication is given by Lemmas~\ref{lem:commutativebroadcasting} and~\ref{lem:BBLWimpliesCBHbroadcasting}.
  The other implication does not hold in $\CPs[\Rel]$; we develop a counterexample.
  Let $G$ be a noncommutative group, considered as a skeletal groupoid. By Lemma~\ref{lem:broadcastableRel}, it is $\otimes$-broadcastable and hence broadcastable. However, it is noncommutative by definition.
\end{proof}

\begin{remark}
  In $\cat{C}=\FHilb$, commutativity and broadcastability are equivalent. They also coincide with a third notion of classicality, namely that a C*-algebra is a direct sum of 1-dimensional C*-algebras. This can be phrased for general compact dagger categories $\cat{C}$ with biproducts, for in that case $\CPs[\cat{C}]$ inherits biproducts~\cite{heunenkissingerselinger:idempotents}. Hence we can consider objects in $\CPs[\cat{C}]$ that arise as biproducts of the monoidal unit. One can straightforwardly show that all such objects give commutative Frobenius algebras, whereas any nontrivial abelian group in $\CPs[\cat{Rel}]$ is commutative, but not a biproduct of units. Hence, we can refine Theorem~\ref{thm:broadcastingvscommutativity} as follows:
  \[
  \text{biproduct of unit}
  \begin{aligned}\begin{array}{c}
    \Rightarrow \\[-1ex]
    \nLeftarrow
  \end{array}\end{aligned}
  \text{commutative}
  \begin{aligned}\begin{array}{c}
    \Rightarrow \\[-1ex]
    \nLeftarrow
  \end{array}\end{aligned}
  \text{broadcastable}
  \]
\end{remark}

Finally, since broadcasting coincides with commutativity when $\cat{C} = \cat{FHilb}$, then broadcasting and $\otimes$-broadcasting must also coincide. We leave open the question of whether this is also true for $\cat{Rel}$.

\section{Bit commitment}\label{sec:commitment}

Briefly, \emph{bit commitment} is the following two-party protocol. Alice claims to know something, and Bob wants to verify that Alice indeed has that knowledge, but Alice doesn't want to reveal her secret yet. Let's say the information is a single bit; Bob wants Alice to commit to either `heads' or `tails' now, and wants to be able to verify her committed value later. Alice could cheat by changing the value she committed to later on; if this is impossible the protocol is \emph{binding}. Bob could cheat by learning the value Alice committed to before she is ready to unveil it; if this is impossible the protocol is \emph{concealing}. A \emph{secure} bit commitment protocol is one where neither cheat is possible. Secure bit commitment is possible under the assumption that a pair of quantum systems can only inhabit classically correlated states. Hence, impossibility of secure bit commitment implies the existence of entangled states \cite{cliftonbubhalvorson}. 
We can model it categorically as follows.

\begin{definition}
  A \emph{bit commitment protocol} for a system $C$ and two sub-subsystems $A, B$ of $C$ consists of two states $H,T \colon I \to C$ of $\CPs[\cat{C}]$ which is \emph{concealing} when:
  \begin{equation}\label{eq:concealing}
    \begin{tikzpicture}[scale=.66]
	\begin{pgfonlayer}{nodelayer}
		\node [style=square box, minimum width=8mm] (0) at (-1.5, -1) {$H$};
		\node [style=square box, minimum width=8mm] (1) at (1.5, -1) {$T$};
		\node [style=none] (2) at (0, 0) {$=$};
		\node [style=none] (3) at (-1.5, 1) {};
		\node [style=none] (4) at (1.5, 1) {};
		\node [style=none] (5) at (-2, 0.75) {$B$};
		\node [style=none] (6) at (-2, -0.25) {$C$};
		\node [style=none] (7) at (2, 0.75) {$B$};
		\node [style=none] (8) at (2, -0.25) {$C$};
	\end{pgfonlayer}
	\begin{pgfonlayer}{edgelayer}
		\draw [style=subsystem] (3.center) to (0);
		\draw [style=subsystem] (4.center) to (1);
	\end{pgfonlayer}
\end{tikzpicture}
  \end{equation}
  and \emph{binding} when there exists no state $\cheat : I \to C$ and morphisms $c_H, c_T$ localised to $A$:
  \[ \begin{tikzpicture}[dotpic]
	\begin{pgfonlayer}{nodelayer}
		\node [style=none] (0) at (-2, 0) {};
		\node [style=none] (1) at (0, 0) {};
		\node [style=white dot] (2) at (-1, 1.75) {};
		\node [style=none] (3) at (-1, 3) {};
		\node [style=none] (4) at (-1, -3.5) {};
		\node [style=none] (5) at (-2, 1) {};
		\node [style=none] (6) at (0, 1) {};
		\node [style=square box] (7) at (-2, -1) {$E_*$};
		\node [style=square box] (8) at (0, -1) {$E$};
		\node [style=none] (9) at (0.5, 0) {$A$};
		\node [style=none] (10) at (-2.5, 0) {$A$};
		\node [style=none] (11) at (-2, -0.5) {};
		\node [style=none] (12) at (0, -0.5) {};
		\node [style=none] (13) at (-0.5, -3) {$C$};
		\node [style=none] (14) at (-0.5, 2.5) {$C$};
		\node [style=square box] (15) at (-5, 0) {$c_H$};
		\node [style=none] (16) at (-5, -1.75) {};
		\node [style=none] (17) at (-5, 1.75) {};
		\node [style=none] (18) at (-3.5, 0) {$:=$};
		\node [style=none] (19) at (-4.5, -1.25) {$C$};
		\node [style=none] (20) at (-4.5, 1.25) {$C$};
	\end{pgfonlayer}
	\begin{pgfonlayer}{edgelayer}
		\draw (4.center) to (2);
		\draw [style=diredge] (2) to (3.center);
		\draw [style=subsystem] (1.center) to (6.center);
		\draw [style=diredge, in=-30, out=90, looseness=0.75] (6.center) to (2);
		\draw [style=none, in=90, out=-150] (2) to (5.center);
		\draw [style=subsystem] (0.center) to (5.center);
		\draw [style=diredge, cross, in=-90, out=-90, looseness=1.50] (7) to (8);
		\draw (11.center) to (0.center);
		\draw (12.center) to (1.center);
		\draw [style=diredge] (16.center) to (15);
		\draw [style=diredge] (15) to (17.center);
	\end{pgfonlayer}
\end{tikzpicture} \qquad\qquad\qquad\qquad \begin{tikzpicture}[dotpic]
	\begin{pgfonlayer}{nodelayer}
		\node [style=none] (0) at (-1, 0) {};
		\node [style=none] (1) at (1, 0) {};
		\node [style=white dot] (2) at (0, 1.75) {};
		\node [style=none] (3) at (0, 3) {};
		\node [style=none] (4) at (0, -3.5) {};
		\node [style=none] (5) at (-1, 1) {};
		\node [style=none] (6) at (1, 1) {};
		\node [style=square box] (7) at (-1, -1) {$F_*$};
		\node [style=square box] (8) at (1, -1) {$F$};
		\node [style=none] (9) at (1.5, 0) {$A$};
		\node [style=none] (10) at (-1.5, 0) {$A$};
		\node [style=none] (11) at (-1, -0.5) {};
		\node [style=none] (12) at (1, -0.5) {};
		\node [style=none] (13) at (0.5, -3) {$C$};
		\node [style=none] (14) at (0.5, 2.5) {$C$};
		\node [style=square box] (15) at (-5, 0) {$c_T$};
		\node [style=none] (16) at (-5, -1.75) {};
		\node [style=none] (17) at (-5, 1.75) {};
		\node [style=none] (18) at (-3, 0) {$=$};
		\node [style=none] (19) at (-4.5, -1.25) {$C$};
		\node [style=none] (20) at (-4.5, 1.25) {$C$};
	\end{pgfonlayer}
	\begin{pgfonlayer}{edgelayer}
		\draw (4.center) to (2);
		\draw [style=diredge] (2) to (3.center);
		\draw [style=subsystem] (1.center) to (6.center);
		\draw [style=diredge, in=-30, out=90, looseness=0.75] (6.center) to (2);
		\draw [style=none, in=90, out=-150] (2) to (5.center);
		\draw [style=subsystem] (0.center) to (5.center);
		\draw [style=diredge, cross, in=-90, out=-90, looseness=1.50] (7) to (8);
		\draw (11.center) to (0.center);
		\draw (12.center) to (1.center);
		\draw [style=diredge] (16.center) to (15);
		\draw [style=diredge] (15) to (17.center);
	\end{pgfonlayer}
\end{tikzpicture} \]
  such that $c_H \circ \cheat = H$ and $c_T \circ \cheat = T$.
  Finally, it is \emph{secure} when it is concealing, and binding.
\end{definition}

This describes a bit commitment protocol in the following way. We assume that initially Alice has access to sub-systems $A$ and $B$, which enable her to prepare the states $H$ and $T$ of $C$. During the commitment phase of the protocol, Alice prepares either $H$ or $T$ and gives subsystem $B$ to Bob. Concealing says that at this point, Bob is unable to determine Alice's commitment. Then, for the reveal phase, Alice gives the remainder of $C$ to Bob, at which point he can ascertain whether he has state $H$ or $T$. Binding says that if Alice only has access to her own system $A$ (namely, after sending system $B$ to Bob), there is no way for Alice to change her commitment.

One thing to note here is that we have somewhat substantially limited Alice's resources if she wishes to cheat, in that she is not allowed to start with a state $\cheat$ which initially occupies a larger system, unbeknownst to Bob, and only send some part of that system in the reveal phase. This allows us to consider situations where secure bit commitment \textit{is} possible, namely when Alice is unable to perform an attack using entanglement.

\begin{example}
  In $\CPs[\FHilb]$, fix the systems $B = (\mathbb C^2, \whitemult)$ and $Q = ((\mathbb C^2)^* \otimes \mathbb C^2, \pantsalg)$ consisting of a bit and qubit. Then the states
  \begin{align*}
  H & = \frac12 \left(\ket{0} \otimes \ketbra{0}{0} + \ket{1} \otimes \ketbra{1}{1} \right) \\
  T & = \frac12 \left(\ket{0} \otimes \ketbra{+}{+} + \ket{1} \otimes \ketbra{-}{-} \right)
  \end{align*}
  give a (na\"ive) secure bit commitment protocol for $(B \otimes Q, B, Q)$. It is concealing because deleting the left system yields the maximally mixed state:
  \[ (\bra{+} \otimes 1) \circ H = (\bra{+} \otimes 1) \circ T = \frac12 \left(\ketbra{0}{0} + \ketbra{1}{1} \right) \]
  and binding because every local operation on a generic state of the form
  \[ p \ket{0} \otimes \rho_0 + q \ket{1} \otimes \rho_1 \]
  yields a state
  \[ p' \ket{0} \otimes \rho_0' + q' \ket{1} \otimes \rho_1' \]
  where $\rho_i'$ is a mixture of the states $\rho_0, \rho_1$. Hence no fixed state \cheat{} can yield both $H$ and $T$ under local operations.
\end{example}

We now translate the result of \cite{cliftonbubhalvorson} into our language. Note that we say an object in a \CPs-category is \textit{noncommutative} if its associated Frobenius structure is noncommutative.

\begin{theorem}
  For a system $C$ in $\CPs[\cat{FHilb}]$ with subsystems $A, B$ such that $B$ is noncommutative, the impossibility of bit commitment entails the existence of an entangled state.
\end{theorem}
\begin{proof}[Proof(sketch)]
    It suffices to show that there exist distinct states $\rho_1, \rho_2$ of $B$ such that:
    \ctikzfig{concealing-hilb}
    yet there exists no classically correlated state and localisable maps $c_1, c_2$ such that $c_1 \circ \rho = \rho_1$ and $c_2 \circ \rho = \rho_2$. This is indeed the case for any noncommutative C*-algebra. Hence, if all states $\rho$ of $C$ are classically correlated, this gives the data of a secure bit commitment protocol, and contrapositively, the impossibility of secure bit commitment necessitates the existence of an entangled state.
\end{proof}

In \cite{cliftonbubhalvorson}, the authors conjectured that the converse is also true, namely that the existence of an entangled state implies the impossibility of secure bit commitment. We now show that, at least in the more general setting, this is not the case.

\begin{proposition}\label{prop:bitcommitment}
  The following defines a secure bit commitment protocol in $\Rel$. Let $\cat{A}$ and $\cat{B}$ be indiscrete groupoids on 2 and 3 objects, respectively. For convenience, we name their morphisms as follows:
  \begin{align*}
    \cat{A} \quad & = \qquad 
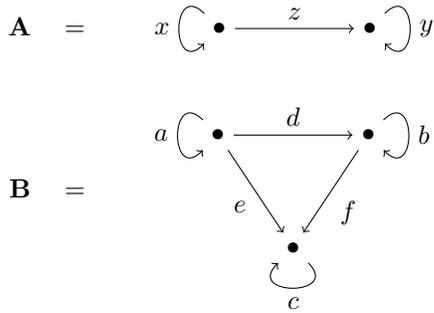
\begin{tikzpicture}
      \node (l) at (0,0) {$\bullet$};
      \node (r) at (2,0) {$\bullet$};
      \draw[->] (l) to node[auto]{$z$} (r);
      \draw[->] (l) to[out=135,in=-135,looseness=4] node[left]{$x$} (l);
      \draw[->] (r) to[out=45,in=-45,looseness=4] node[right]{$y$} (r);
    \end{tikzpicture} \\
    \cat{B} \quad & = \qquad \begin{tikzpicture}
      \node (l) at (0,0.75) {$\bullet$};
      \node (r) at (2,0.75) {$\bullet$};
      \node (d) at (1,-0.75) {$\bullet$};
      \draw[->] (l) to[out=135,in=-135,looseness=4] node[left]{$a$} (l);
      \draw[->] (r) to[out=45,in=-45,looseness=4] node[right]{$b$} (r);
      \draw[->] (d) to[out=-45,in=-135,looseness=4] node[below]{$c$} (d);
      \draw[->] (l) to node[auto]{$d$} (r);
      \draw[->] (l) to node[auto,swap]{$e$} (d);
      \draw[->] (r) to node[auto]{$f$} (d);
    \end{tikzpicture}
  \end{align*}
  (note the inverses $d^{-1}$, $e^{-1}$, $f^{-1}$, and $z^{-1}$ are not pictured).
  Let $\cat{C} = \cat{A} \times \cat{B}$, and define the following states of $\cat{C}$:
  \begin{align*}
    H & = (\{x\} \times \cat{B}) \cup \big\{(y,b) ,(z,d) ,(z^{-1},d^{-1})  \big\}, \\
    T & = (\{x\} \times \cat{B}) \cup \big\{(y,c) ,(z,e) ,(z^{-1},e^{-1})  \big\}\text.
  \end{align*}
  Then, $(\cat{A}, \cat{B}, H, T)$ gives a bit commitment protocol for $\cat{C}$.
\end{proposition}
\begin{proof}
  The key to this counterexample is that $H$ and $T$ are similar enough to be concealing, but still distinct enough that Alice cannot possibly find local maps $R$ and $S$ which produce them from a single state.

  $H, T$ are closed under taking domains and inverses, so they are indeed states in $\CPs[\Rel]$. To see that this is concealing, note that $(\whitecounit \idpic) \circ H = (\whitecounit \idpic) \circ T = \cat{B}$. To see that this protocol is binding, suppose there exists a state $\cheat$ and maps $R, S \colon \cat{A} \to \cat{A}$ such that both $(R \otimes 1_{\cat{B}}) \circ \cheat = H$ and $(S \otimes 1_{\cat{B}}) \circ \cheat = T$. For a morphism $p \in \cat{A}$, let $\langle p \rangle \colon \cat{A} \to \{*\}$ be the (not necessarily inverse-respecting) relation $\{(p,*)\}$. Then, for each of the following values for $\pi \colon \cat{A} \to \{*\}$:
  \[ 
    \langle y \rangle \circ R, 
    \langle z \rangle \circ R, 
    \langle {z^{-1}} \rangle \circ R, 
    \langle y \rangle \circ S, 
    \langle z \rangle \circ S, 
    \langle {z^{-1}} \rangle \circ S, 
  \]
  the composition $(\pi \otimes 1_{\cat{B}}) \circ \cheat$ yields a distinct singleton subset of $\cat{B}$ (namely: $\{b\}$, $\{d\}$, $\{d^{-1}\}$, $\{c\}$, $\{e\}$, and $\{e^{-1}\}$, respectively). Hence each of these 6 effects, regarded as subsets of $\cat{A}$, must contain at least one element which is not in the other 5. Since $\cat{A}$ has only 4 morphisms, this yields a contradiction.
\end{proof}

\begin{definition}\label{def:entangled}
  Let $\cat{C}$ be a compact dagger category, and $C$ an object in $\CPs[\cat{C}]$.
  A state $\rho \colon I \to C$ is \emph{classically correlated} with respect to subsystems $A$ and $B$ of $C$ when there exist a broadcastable object $X$ in $\CPs[\cat{C}]$, a state $p \colon I \to X$, and a morphism $f \colon X \to C$ satisfying
  $\rho = f \circ p$ and:
  \ctikzfig{classicallycorrelated}
  We say that $\rho$ is \emph{entangled} when it is not classically correlated.
\end{definition}

\begin{lemma}\label{lem:nonlocal:rel}
  Let $\cat A$ and $\cat B$ be wide subgroupoids of $\cat C$, regarded as subsystems in $\CPs[\cat{Rel}]$, and let $R \colon 1 \to \cat C$ be a morphism, regarded as a subset $R \subseteq \cat C$. Suppose that $R$ classically correlates $\cat A$ and $\cat B$.
  If $a \circ b \in R$, then $\id[\dom(a)] \in R$.
\end{lemma}
\begin{proof}
  Definition~\ref{def:entangled} unfolds as follows: there is a skeletal $G$ and maps $S \colon G \to C$ and $P \colon I \to G$ with $R = P \circ S$ and such that $(g,a \circ b) \in S$ if and only if $(h,a) \in S$ and $(k,b) \in S$ and $g = h \circ k$ for some $h,k \in G$.

  If $a \circ b \in R$, then there is $g \in P$ with $(g, a\circ b) \in S$. 
  So there are $g,h \in G$ with $g \circ h \in P$ and $(g,a) \in S$ and $(b,h) \in S$. 
  Write $e$ for identity on $\dom(g)=\cod(h)$; because $G$ is skeletal also $\dom(h)=\cod(h)$ and hence $e \in P$. 
  Then $(e,\id[\dom(a)]) \in S$, and so $\id[\dom(a)] \in R$.
\end{proof}

\begin{remark}\label{rem:nonlocal:rel}
  To illustrate the slight abuse of notation in Example~\ref{lem:subsystem:rel}, let us show that in case $\cat{C}=\cat{A} \times \cat{B}$ the subsystems $\cat{A}$ and $\cat{B}$ are in fact multi-valued. In the setting of the previous lemma, we will show that if $(a,b) \in R$, then $(\id[\dom(a)],\id[\cod(b)]) \in R$. 

  Because $C=A \times B$, Definition~\ref{def:entangled} simplifies as follows: we may regard $R$ as a map $A \to B$; there is a skeletal groupoid $G$ and maps $S \colon G \to A \times B$ and $P \colon I \to G$ with:
  \ctikzfig{entangled_rel}
  If $(a,b) \in R$ then there is $g \in P$ with $(g,(a,b)) \in S$. 
  Hence there are $g,h \in G$ with $g \circ h \in P$ and $(g,(a,\id[y])) \in S$ for some object $y$ of $B$ and $(h,(\id[x],b)) \in S$ for some object $x$ of $A$. 
  Again the identity $e = e \circ e$ on $\dom(g)=\cod(h)$ is in $P$, and $(e,(\id[\dom(a)],\id[y])) \in S$ and $(e,(\id[x],\id[\cod(b)])) \in S$. Therefore $(\id[\dom(a)],\id[\cod(b)]) \in R$.
\end{remark}

\begin{theorem}\label{thm:bitcommitmentvslocality}
  In $\CPs[\cat{Rel}]$:
  \[
  \text{no bit commitment}
  \nLeftarrow
  \text{entanglement}
  \]
\end{theorem}
\begin{proof}
  Proposition~\ref{prop:bitcommitment} exhibited a secure bit commitment protocol in $\CPs[\cat{Rel}]$. It suffices to show that the system $\cat{C}$ defined there admits an entangled state. Let
  \[ E = \big\{ (x, a), (y,c), (z,e), (z^{-1}, e^{-1}) \big\} \]
  be a state of $\cat{C}$. Then, note that $(z, e) = (y, f) \circ (z, d)$. However, $\id[\dom(y,f)] = (y, b) \notin E$. Hence by Lemma~\ref{lem:nonlocal:rel} (and Remark~\ref{lem:nonlocal:rel}), $E$ is entangled.
\end{proof}

The point this paper emphasises is that the impossibility of secure bit commitment is not caused by the conceptual structure of quantum theory, but by the algebraic model assumed in~\cite{cliftonbubhalvorson}.

\bibliographystyle{plain}
\bibliography{cbh}

\begin{thebibliography}{10}

\bibitem{abramskycoecke:categoricalsemantics}
S.~Abramsky and B.~Coecke.
\newblock A categorical semantics of quantum protocols.
\newblock In {\em Logic in Computer Science 19}, pages 415--425. IEEE Computer
  Society, 2004.

\bibitem{barnumetal:nobroadcasting}
H.~Barnum, J.~Barrett, M.~Leifer, and A.~Wilce.
\newblock A generalized no-broadcasting theorem.
\newblock {\em Physical Review Letters}, 99:240501, 2007.

\bibitem{NoBroadcasting}
H.~Barnum, C.~M. Caves, C.~A. Fuchs, R.~Jozsa, and B.~Schumacher.
\newblock Noncommuting mixed states cannot be broadcast.
\newblock {\em Phys. Rev. Lett.}, 76:2818--2821, Apr 1996.

\bibitem{barrett:gpt}
J.~Barrett.
\newblock Information processing in generalized probabilistic theories.
\newblock {\em Physical Review A}, 75(3):032304, 2007.

\bibitem{brown2006topology}
R.~Brown.
\newblock {\em Topology and Groupoids}.
\newblock www.groupoids.org, 2006.

\bibitem{Chiri2}
G.~Chiribella, G.~M. D'Ariano, and P.~Perinotti.
\newblock Informational derivation of quantum theory.
\newblock {\em Physical Review A}, 84(1):012311, 2011.

\bibitem{choi:completelypositive}
M.-D. Choi.
\newblock Completely positive linear maps on complex matrices.
\newblock {\em Linear Algebra and its Applications}, 10:285--290, 1975.

\bibitem{cliftonbubhalvorson}
R.~Clifton, J.~Bub, and H.~Halvorson.
\newblock Characterizing quantum theory in terms of information-theoretic
  constraints.
\newblock {\em Foundations of Physics}, 33(11):1561--1591, 2003.

\bibitem{coecke2016terminality}
B.~Coecke.
\newblock Terminality implies no-signalling... and much more than that.
\newblock {\em New Generation Computing}, 34(1-2):69--85, 2016.

\bibitem{CD2}
B.~Coecke and R.~Duncan.
\newblock Interacting quantum observables: categorical algebra and
  diagrammatics.
\newblock {\em New Journal of Physics}, 13:043016, 2011.
\newblock {arXiv:quant-ph/09064725}.

\bibitem{CompQL}
B.~Coecke, C.~Heunen, and A.~Kissinger.
\newblock {\em Computation, Logic, Games, and Quantum Foundations. The Many
  Facets of Samson Abramsky: Essays Dedicated to Samson Abramsky on the
  Occasion of His 60th Birthday}, chapter Compositional Quantum Logic, pages
  21--36.
\newblock Springer Berlin Heidelberg, Berlin, Heidelberg, 2013.

\bibitem{coeckeheunenkissinger:cpstar}
B.~Coecke, C.~Heunen, and A.~Kissinger.
\newblock Categories of quantum and classical channels.
\newblock {\em Quantum Information Processing}, 2014.

\bibitem{CQMII}
B.~Coecke and A.~Kissinger.
\newblock Categorical quantum mechanics {II}: classical-quantum interaction.
\newblock ar{X}iv:1605.08617 [quant-ph], 2016.

\bibitem{PQP}
B.~Coecke and A.~Kissinger.
\newblock {\em Picturing Quantum Processes: A First Course in Quantum Theory
  and Diagrammatic Reasoning}.
\newblock Cambridge University Press, 2017.

\bibitem{coeckelalcausal}
B.~Coecke and R.~Lal.
\newblock Causal categories: relativistically interacting processes.
\newblock {\em Foundations of Physics}, 43(4):458--501, 2013.

\bibitem{coeckepaquette:naimark}
B.~Coecke and {\'E}.~O. Paquette.
\newblock {POVM}s and {N}aimark's theorem.
\newblock In {\em QPL 2006}, volume 210 of {\em Electronic Notes in Theoretical
  Computer Science}, pages 15--31, 2006.

\bibitem{cohngordon:bitcommitment}
K.~Cohn-Gordon.
\newblock Commitment algorithms.
\newblock Master's thesis, University of Oxford, 2012.

\bibitem{fritz2009possibilistic}
T.~Fritz.
\newblock {\em Possibilistic Physics}.
\newblock 2009.

\bibitem{gachechiladze:heisenberg}
M.~Gachechiladze.
\newblock On categorical characterizations of no-signaling theories.
\newblock Master's thesis, University of Oxford, 2014.

\bibitem{hardy:reconstruction}
L.~Hardy.
\newblock {\em Quantum theory: informational foundations and foils}, chapter
  Reconstructing quantum theory, pages 223--248.
\newblock Springer, 2016.

\bibitem{heunencontrerascattaneo:groupoids}
C.~Heunen, I.~Contreras, and A.~S. Cattaneo.
\newblock Relative {F}robenius algebras are groupoids.
\newblock {\em Journal of Pure and Applied Algebra}, 217(1):114--124, 2013.

\bibitem{heunenkissingerselinger:idempotents}
C.~Heunen, A.~Kissinger, and P.~Selinger.
\newblock Completely positive projections and biproducts.
\newblock In {\em QPL 2014}, volume 171 of {\em Electronic Notes in Theoretical
  Computer Science}, pages 71--83, 2014.

\bibitem{heunenvicary:cqt}
C.~Heunen and J.~Vicary.
\newblock {\em Categories for Quantum Theory: An Introduction}.
\newblock Oxford University Press, 2019.

\bibitem{hoehn2015quantum}
P.~A. Hoehn and C.~Wever.
\newblock Quantum theory from questions.
\newblock {\em arXiv preprint arXiv:1511.01130}, 2015.

\bibitem{keyl2002fundamentals}
Michael Keyl.
\newblock Fundamentals of quantum information theory.
\newblock {\em Physics reports}, 369(5):431--548, 2002.

\bibitem{maclane:categories}
S.~{Mac Lane}.
\newblock {\em Categories for the Working Mathematician}.
\newblock Springer, 2nd edition, 1971.

\bibitem{masanes2011derivation}
L.~Masanes and M.~P. M{\"u}ller.
\newblock A derivation of quantum theory from physical requirements.
\newblock {\em New Journal of Physics}, 13(6):063001, 2011.

\bibitem{masanes2013existence}
L.~Masanes, M.~P. M{\"u}ller, R.~Augusiak, and D.~P{\'e}rez-Garc{\'\i}a.
\newblock Existence of an information unit as a postulate of quantum theory.
\newblock {\em Proceedings of the National Academy of Sciences},
  110(41):16373--16377, 2013.

\bibitem{QuickThesis}
D.~Quick.
\newblock {\em {!-Logic: First Order Reasoning for Families of Non-Commutative
  String Diagrams}}.
\newblock PhD thesis, Oxford, 2015.

\bibitem{schlosshauer:enigma}
M.~Schlosshauer.
\newblock {\em Elegance and Enigma}.
\newblock Springer, 2011.

\bibitem{schumacher2012modal}
Benjamin Schumacher and Michael~D Westmoreland.
\newblock Modal quantum theory.
\newblock {\em Foundations of Physics}, 42(7):918--925, 2012.

\bibitem{SelingerCPM}
P.~Selinger.
\newblock Dagger compact closed categories and completely positive maps.
\newblock {\em Electronic Notes in Theoretical Computer Science}, 170:139--163,
  2007.

\bibitem{SelingerSurvey}
P.~Selinger.
\newblock A survey of graphical languages for monoidal categories.
\newblock In B.~Coecke, editor, {\em New Structures for Physics}, Lecture Notes
  in Physics, pages 275--337. Springer-Verlag, 2011.
\newblock {a}rXiv:0908.3347.

\bibitem{spekkens:toytheory}
R.~W. Spekkens.
\newblock Evidence for the epistemic view of quantum states: a toy theory.
\newblock {\em Physical Review A}, 75(3):032110, 2007.

\bibitem{Vicary2011}
J.~Vicary.
\newblock Categorical formulation of finite-dimensional {C*}-algebras.
\newblock {\em Electronic Notes in Theoretical Computer Science}, 270(1):129 --
  145, 2011.
\newblock Proceedings of the Joint 5th International Workshop on Quantum
  Physics and Logic and 4th Workshop on Developments in Computational Models
  (QPL/DCM 2008).

\end{thebibliography}

\end{document}